\documentclass[11pt,usletter]{article}
\usepackage{jheppub}
\pdfoutput=1

\usepackage{amsthm}
\usepackage{url}
\usepackage{comment}
\usepackage{physics}
\usepackage{graphicx}
\usepackage{hyperref}
\usepackage{tikz-cd}
\usepackage{color}

\def\ba{\begin{align}}
\def\ea{\end{align}}
\def\be{\begin{equation}}
\def\ee{\end{equation}}
\def\bea{\begin{eqnarray}}
\def\eea{\end{eqnarray}}

\theoremstyle{plain}
\newtheorem{thm}{Theorem}[section]
\newtheorem{lem}[thm]{Lemma}

\theoremstyle{definition}

\theoremstyle{remark}

\hbox{CALT-TH-2019-006}
\begin{document}

\title{Note on global symmetry and SYK model}

\date{\today}

\author[a,b]{Junyu Liu,}
\author[c,d]{and Yehao Zhou}

\affiliation[a]{Walter Burke Institute for Theoretical Physics,\\ California Institute of Technology, Pasadena, California 91125, USA}
\affiliation[b]{Institute for Quantum Information and Matter, California Institute of Technology, Pasadena, CA 91125, USA}
\affiliation[c]{Perimeter Institute for Theoretical Physics,\\ Waterloo, ON N2L 2Y5, Canada}
\affiliation[d]{Department of Physics \& Astronomy, University of Waterloo,\\ Waterloo, ON N2L 3G1, Canada}

\emailAdd{jliu2@caltech.edu}
\emailAdd{yzhou3@perimeterinstitute.ca}

\abstract{The goal of this note is to explore the behavior of effective action in the SYK model with general continuous global symmetries. A global symmetry will decompose the whole Hamiltonian of a many-body system to several single charge sectors. For the SYK model, the effective action near the saddle point is given as the free product of the Schwarzian action part and the free action of the group element moving in the group manifold. With a detailed analysis in the free sigma model, we prove a modified version of Peter-Weyl theorem that works for generic spin structure. As a conclusion, we could make a comparison between the thermodynamics and the spectral form factors between the whole theory and the single charge sector, to make predictions on the SYK model and see how symmetry affects the chaotic behavior in certain timescales.
}


\maketitle
\flushbottom

\section{Introduction}
For an isolated quantum statistical system, we will assign a time-independent Hamiltonian $H$. If we have a global symmetry $Q$, then we have
\begin{align}
\left[H,Q\right]=0
\end{align}
thus, $Q$ will help to decompose the eigenspace of Hamiltonian into multiple subspaces, characterized by different eigenvalues of $Q$. If we call $Q$ as the charge operator, different eigenspaces could be called as charge sectors.

If the system is in a thermal bath with inverse temperature $\beta$, one could assign a partition function
\begin{align}
Z(\beta ) = {\rm{Tr}}\left( {{e^{ - \beta H}}} \right)
\end{align}
However, if the system also has a fixed chemical potential $\phi$\footnote{We follow the notation \cite{Stanford:2017thb} here, while in another convention we may not introduce the imaginary unit $i$. The notation we use will be convenient when performing transformation between canonical and grand canonical ensemble, where we could use Fourier transform instead of Laplace transform.}, one could study the system as a grand canonical ensemble
\begin{align}
Z(\beta|\phi ) = {\rm{Tr}}\left( {{e^{ - \beta H + i\phi Q}}} \right)
\end{align}
One can transform the partition function, from grand canonical ensemble with fixed $\beta$ and $\phi$, to canonical ensemble with fixed $\beta$ and $\mu$. Here, a fixed $\mu$ means that we are only considering the subspace where $Q$'s eigenvalue is restricted to $\mu$. Namely, we are addressing the partition function in a single charge sector. 

The roles of $\mu$ and $\phi$ are simply related by Fourier conjugations. Let us take $\text{U}(1)$ charge as an example, where we have
\begin{align}
Z(\beta ,\mu) = \int_0^{2\pi } {\frac{{d\phi }}{{2\pi }}{e^{ - i\phi \mu}}} Z(\beta|\phi )
\end{align}
In this paper, we will apply this basic knowledge to a specific system, the Sachdev-Ye-Kitaev(SYK) model associated with a global, continues symmetry $G$.  The discovery of SYK model \cite{Sachdev:1992fk,A1,A2,A3,A4,Maldacena:2016hyu} opens a novel research direction towards quantum chaos in quantum gravity (see also \cite{Almheiri:2014cka,Fu:2016yrv,Polchinski:2016xgd,Jevicki:2016bwu,Jensen:2016pah,Maldacena:2016upp,Fu:2016vas,Witten:2016iux,Davison:2016ngz}). The SYK model is nearly conformal and maximal chaotic in certain limit, which is conjectured to reflect some features in the near $\text{AdS}_2$ gravity and black hole physics. A concrete and complete study of this model and related generalizations (for instance, supersymmetric generalization in \cite{Fu:2016vas}) is believed to provide some mysterious features in quantum gravity and holography. 

In SYK model (or more generally, its various generalizations), the $1/N$ fluctuations above the saddle point solution is captured by an effective action. The action has a Schwarzian derivative 
\begin{align}
{S_\varphi }[\varphi ] =  - \frac{1}{{g_c^2}}\int_0^{2\pi } {d\tau {\rm{Sch}}\left\{ {\tan \frac{\varphi }{2},\tau } \right\}} 
\end{align}
where $\varphi(\tau)$ is the fluctuation field above the saddle point, where the Schwarzian derivative is defined as
\begin{align}
{\rm{Sch}}\left( {f,z} \right) = \frac{{f'''(z)}}{{f'(z)}} - \frac{3}{2}{\left( {\frac{{f''(z)}}{{f'(z)}}} \right)^2}
\end{align}
and $1/g_c^2$ is the coupling scales as $N/\beta J$, where $J$ is the randomness of the model, $\beta$ is the inverse temperature, and $N$ is number of fermions in the model.

The Schwarzian action is significant in the sense of giving a maximal chaotic exponent. One can evaluate the partition function of this action and compute thermodynamical variables in the low energy limit by computing the one loop determinant. The discovery of Stanford-Witten localization \cite{Stanford:2017thb} shows that this partition function is one-loop exact. Namely, one can trust the calculation even in the strong coupling case. The dependence with the temperature in the one loop partition function could be obtained in the density of states, and also the spectral form factor \cite{You:2016ldz,Cotler:2016fpe,Cotler:2017jue,Liu:2018hlr,Saad:2018bqo}.

In this paper, we are interested in the case where the SYK model is associated with a global symmetry $G$. In the simplest case, we will discuss a direct product between the Schwarzian part and the sigma model corresponding to the group $G$. The structure of the action is discussed briefly in Section 3.2 in \cite{Stanford:2017thb}. In this case, the whole action is written as a sum of the Schwarzian part and the phase field capturing the global symmetry $G$, which is free moving in the group manifold and is known to be one-loop exact \cite{Marinov:1979gm,Picken:1988ev,Chu:1994hm}. Because it is a direct product among two manifolds when performing the path integral, the one-loop localization follows trivially and rigorously for the whole theory. 

One famous example for the symmetrized SYK model is the complex SYK model (the Sachdev-Ye model, see \cite{Sachdev:1992fk,Fu:2016vas}, and some recent studies \cite{Bulycheva:2017uqj,Chaturvedi:2018uov,Bhattacharya:2018nrw}), where $G=\text{U}(1)$. In this theory, the free action on $\text{U}(1)$ \cite{pc,appear,Sachdev:2019bjn,pc2} will contribute some extra effects to the Schwarzian part, and create a new temperature dependence \cite{Davison:2016ngz}. One can also construct some more generic Lie groups \cite{Gross:2016kjj,Yoon:2017nig,Narayan:2017hvh}, and various global symmetries may also appear in tensor model, an analog of SYK without disorder but with similar large $N$ dynamics (see \cite{Witten:2016iux,Klebanov:2016xxf,Klebanov:2017nlk,Giombi:2017dtl,Bulycheva:2017ilt,Klebanov:2018nfp,Giombi:2018qgp,Pakrouski:2018jcc} for reference). 

Why we need extra symmetries? The original SYK model is very successful in the sense of maximal chaos, and capture part of near horizon physics in $\text{AdS}_2$. However, such a quantum mechanical example in one dimension is special, and it is not completely clear how a full semi-classical dual theory with all sectors including gravitational sector, should emerge from such a theory. On the one hand, deeper insights in near $\text{AdS}_2$ geometry are needed, but on the other hand, one may consider constructing and studying deeper alternative models that could capture features we learn from SYK (for instance,  easy to solve and maximal chaotic), and could have a more clear dual picture and work for higher dimensions. Sometimes, extra symmetries are hard to avoid in those generalizations, due to more complicated symmetries of the dual black hole horizon we need, or due to our current limited understanding about holography (for instance, supersymmetry). Moreover, other symmetries may lead to interpretations of symmetry and charge in gravity, and some previous discussions about Kerr/CFT \cite{Guica:2008mu,Anninos:2017cnw}. Although we didn't study any specific models in this paper, we interpret the current study as the first step towards more detailed features among the current and future SYK-like models. Moreover, this paper only studies global continuous symmetries. For discrete symmetries (in general chaotic systems see \cite{thesis}) and supersymmetry (for instance, see \cite{Li:2017hdt,Stanford:2017thb,Kanazawa:2017dpd,Hunter-Jones:2017crg}), it might be valuable to study deeper following similar spirit, using technologies from condensed matter physics, quantum information and quantum gravity (for recent discussions about symmetries in quantum field theory and quantum gravity, see \cite{Harlow:2018tng,Harlow:2018jwu,Harlowtalk,toa}. For quantum circuits and black hole thought experiments with $\text{U}(1)$ symmetry, see \cite{Rakovszky:2017qit,Khemani:2017nda,beni}).

In this paper, we will systematically discuss the symmetrized effective action, keeping $G$ to be general. We will compute the expressions $Z(\beta,\mu)$ and $Z(\beta|\phi)$ explicitly with various examples (as a summary, see Section \ref{fe}), and study their predictions on the chaotic and thermodynamical observables (as a sketch, see Table \ref{tb1} and \ref{tb2}).

The paper is organized as the following. In Section \ref{ref} we will provide a simple review of the SYK model and their generalizations. In Section \ref{partition} we will discuss the computation of the partition function mostly in the free theory in detail. In Section \ref{pred} we will discuss predictions in SYK-like models. In Section \ref{conc} we arrive at a conclusion and discussion.

\section{About SYK model}\label{ref}

The (majonara) SYK model is a one-dimensional condensed matter model with $N$ majonara fermions. The model has a disorder average over non-local coupling,
\begin{align}
H = \sum\limits_{i < j < k < l} {{J_{ijkl}}{\psi ^i}{\psi ^j}{\psi ^k}{\psi ^l}} 
\end{align}
(we write the four-local case for simplicity as the simplest example) where $i,j,k,l$ is ranging from $1$ to $N$, and $\psi$s are majonara fermions. The coupling is a Gaussian random variable
\begin{align}
\left\langle {{J_{ijkl}}} \right\rangle  = 0~~~~~~\left\langle {J_{ijkl}^2} \right\rangle  = \frac{{6{J^2}}}{{{N^3}}}
\end{align}
where $J$ is a positive constant which sets the scale where the dimensionless coupling is $\beta J$.  In the large $N$ and IR limit $1\ll \beta J\ll N$, one can show that the large $N$ solution of the two point function has the $\text{SL}(2,\mathbb{R})$ covariance. In the strict IR limit, the theory has the reparametrization symmetry ($\text{diff}(S^1)$), so the space of Nambu-Goldstone bosons is $\text{diff}(S^1)/\text{SL}(2,\mathbb{R})$.

One can study the effective field theory for reparametrization mode $\varphi \in \text{diff}(S^1)/\text{SL}(2,\mathbb{R})$. The theory is described by the Schwarzian action 
\begin{align}
{S_\varphi }[\varphi ] =  - \frac{1}{{g_c^2}}\int_0^{2\pi } {d\tau {\rm{Sch}}\left\{ {\tan \frac{\varphi }{2},\tau } \right\}} 
\end{align}
where $\frac{1}{{g_c^2}} = \frac{{2\pi N\alpha }}{{\beta J}}$, and $\alpha$ is a constant that has been computed numerically. The partition function of this action is shown to be one-loop exact. The one-loop partition function is written as 
\begin{align}
{Z_\varphi }(\beta ) \sim \frac{1}{{{{(\beta J)}^{3/2}}}}\exp (\frac{\pi }{{g_c^2}})
\end{align}
The dependence $1/(\beta J)^{3/2}$ determines the speed of scrambling in the observables like analytic continued partition function and the spectral form factor of the theory. (Several recent papers address the study of the Schwarzian action, see \cite{Mertens:2017mtv,Mertens:2018fds,Qi:2018rqm,Lam:2018pvp,Blommaert:2018oro}.)

Now we wish to understand how a global symmetry will change the scaling of the partition function. Although we wish to discuss a general symmetry group, the $\text{U(1)}$ case will be the simplest example. It will show up in the complex SYK model, which is defined as
\begin{align}
{H} = \sum\limits_{1 \le {i_1} < {i_2} \le N} {\sum\limits_{1 \le {i_3} < {i_4} \le N} {{J_{{i_1}.{i_2},{i_3},{i_4}}}f_{{i_1}}^\dag f_{{i_2}}^\dag {f_{{i_3}}}{f_{{i_4}}}} } 
\end{align}
where $f$s are Dirac fermions, and $J$s satisfy 
\begin{align}
{J_{{i_1}.{i_2},{i_3},{i_4}}} = J_{{i_3},{i_4},{i_1}.{i_2}}^*~~~~~~\left\langle {{{\left| {{J_{{i_1}.{i_2},{i_3},{i_4}}}} \right|}^2}} \right\rangle  = \frac{{4{J^2}}}{{{N^3}}}
\end{align}
In this model, there is an $\text{U}(1)$ symmetry and the charge is conserved. One can define the fermionic charge 
\begin{align}
\mathcal{Q}=\frac{1}{N}\sum_i(f_i^\dagger f_i-\frac{1}{2})~~~~~~[\mathcal{Q},H]=0
\end{align}
The paper \cite{Davison:2016ngz} studies the model in detail. Here we will briefly describe its effective field theory. The effective action is written as
\begin{align}
& S={{S}_{\psi }}+{{S}_{\varphi}} \nonumber\\
& {{S}_{\psi}}=\frac{K}{4}\int_{0}^{\beta }{d\tau {{\left( {{\partial }_{\tau }}\tilde{\phi} +\frac{2\pi i\mathcal{E}}{\beta }{{\partial }_{\tau }}\varphi \right)}^{2}}}\nonumber\\
& {{S}_{\varphi }}=-\frac{\gamma }{4{{\pi }^{2}}}\int_{0}^{\beta }{d\tau \text{Sch}\left\{ \tan \frac{\pi \varphi}{\beta },\tau  \right\}} 
\end{align}
where $K$ and $\gamma$ are some thermodynamical quantities which could be computed numerically and they scale as $N/J$. Here we notice that ${{S}_{\varphi }}$ is the same for the Schwarzian action of the majonara SYK model. Here $\varphi(\tau)=\tau+\delta \varphi(\tau)$ is the reparametrization, and $\tilde{\phi}$ is a phase field capturing the $\text{U}(1) $ symmetry, and it has the periodicity $\tilde{\phi}\sim \tilde{\phi}+2\pi$. The constant $\mathcal{E}$ is a thermodynamical quantity that is defined as
\begin{align}
2\pi \mathcal{E}=\frac{d\mathcal{S}(\mathcal{Q})}{d\mathcal{Q}}
\end{align}
where $\mathcal{S}$ is the entropy, and we could define a shift of the field
\begin{align}
\psi =\tilde{\phi}+\frac{2\pi i\mathcal{E}}{\beta }\varphi
\end{align}
So we have
\begin{align}
{{S}_{\psi }}=\frac{K}{4}\int_{0}^{\beta }{d\tau {{\left( {{\partial }_{\tau }}\psi  \right)}^{2}}}
\end{align}
The periodicity for $\psi$ is still $2\pi$. One can send $\tau\to 2\pi \tau/\beta$, such that these integrals become
\begin{align}
  & {{S}_{\psi }}=\frac{K\pi }{2\beta }\int_{0}^{2\pi }{d\tau {{\left( {{\partial }_{\tau }}\psi  \right)}^{2}}} \nonumber\\
 & {{S}_{\varphi }}=-\frac{\gamma }{2\pi \beta }\int_{0}^{2\pi }{d\tau \text{Sch}\left\{ \tan \frac{\varphi}{2},\tau  \right\}} 
\end{align}
For symmetry groups more general than $\text{U}(1)$, models are precisely constructed in, for instance, \cite{Gross:2016kjj,Yoon:2017nig,Narayan:2017hvh}. In those models, the form of the effective action is generic: a Schwarzian mode for reparametrization symmetry, and a phase field moving in a group manifold. 

In this paper, we will study the one-loop partition function given from the following action
\begin{align}
&S = {S_f} + {S_\varphi }\nonumber\\
&{S_f} =  - \frac{{K\pi }}{2\beta}\int_0^{2\pi } {{\rm{Tr}}} {({f^{ - 1}}{\partial _\tau }f)^2}d\tau \nonumber\\
&{S_\varphi } =  - \frac{\gamma }{{2\pi \beta }}\int_0^{2\pi } {d\tau {\rm{Sch}}\left\{ {\tan \frac{\varphi }{2},\tau } \right\}} 
\end{align}
where $f$ is a phase field moving in a generic group $G$. We will study general $G$ with certain assumption: compact semisimple. For non-semisimple case, similar technologies could be used, and we will discuss $\text{U}(M)$ as examples. The goal of us is to understand the partition function generated by the above action, in the grand canonical and canonical ensembles, and to understand their relations, which is highly relying on the classic study of free sigma model moving on a Lie group.

For the range of $\beta$, the validity of the effective action for the SYK-like theory is $\beta J\gg 1$. In this paper, we are mostly interested in two possible ranges, $1\ll \beta J \ll N$, and $\beta J \gg N$ (namely, $K\gg \beta$ or $K\ll \beta$.)

\section{Studying the sigma model}\label{partition}
\subsection{\text{U}(1) as a warmup}
As a pedagogical example, we will start from $\text{U}(1)$ \cite{pc,appear,pc2,Sachdev:2019bjn}. $\text{U}(1)$ is not a semisimple group, and it has two different spin structures. In the complex SYK model, only trivial spin structures would present, while for $\mathcal{N}=2$ supersymmetric SYK model \cite{Fu:2016vas,Stanford:2017thb}, the spin structure depends on if the total number of particles is even or odd.

Using the complex SYK model notation in the previous section, we write down the sigma model for \text{U}(1) as 
 \begin{align}
 {{S}_{\psi }}=\frac{K\pi }{2\beta }\int_{0}^{2\pi }{d\tau {{\left( {{\partial }_{\tau }}\psi  \right)}^{2}}} 
 \end{align}
By solving the equation of motion, we could have infinite number of saddle points 
\begin{align}
{\psi _{\hat{n}}} =\hat{n}\tau 
\end{align}
with the corresponding action
\begin{align}
{S_\psi } = \frac{{K\pi }}{2\beta }\int_0^{2\pi } {d\tau {{\left( {{\partial _\tau }(\hat{n}\tau )} \right)}^2}}  = \frac{{{\hat{n}^2}{\pi ^2}K}}{\beta } 
\end{align}
Now we start to compute the one-loop partition function. We study the perturbation around the saddle point
\begin{align}
  & {{\psi }_{\hat{n}}}=\hat{n}\tau +\delta \psi  \nonumber\\
 & \delta \psi =\sum\limits_{\hat{p}}{{{\psi }_{\hat{p}}}{{e}^{i\hat{p}\tau }}} 
\end{align}
Thus we get
\begin{align}
& \frac{K\pi }{2\beta }\int_{0}^{2\pi }{d\tau {{\left( {{\partial }_{\tau }}\delta \psi  \right)}^{2}}} =-\frac{K{{\pi }^{2}}}{\beta }\sum\limits_{\hat{p}\in \mathbb{Z}}{{{\hat{p}}^{2}}{{\psi }_{\hat{p}}}{{\psi }_{-\hat{p}}}} 
\end{align}
Using the zeta-function regularization, and cutting out the zero mode $p=0$, we would get
\begin{align}
\log {{Z}_{\psi }}\sim -\sum\limits_{\hat{p}\in {{\mathbb{Z}}^{+}}}{\log \left( -\frac{K{{\pi }^{2}}}{\beta }{{\hat{p}}^{2}} \right)}\sim -\frac{1}{2}\log (\frac{K}{\beta})
\end{align}
Thus, a single saddle point parameterized by $\hat{n}$ will contribute the partition function by
\begin{align}
{Z_{\psi ,\hat n}}\sim {\left( {\frac{K}{\beta }} \right)^{1/2}}\exp ( - \frac{{K{\pi ^2}{{\hat n}^2}}}{\beta })
\end{align}
$\text{U}(1)$ has two spin structures: the trivial one $\sigma_0$, and the M\"obius $\sigma_1$. Those correspond to even and odd particles. Using this, we could compute the whole partition function, with zero chemical potential, by
\begin{align}
&{Z_{{\sigma _0}}}(\beta |\phi  = 0) = \sum\limits_{\hat n} {{Z_{\hat n}}}  \sim \sum\limits_{\hat n} {{{\left( {\frac{K}{\beta }} \right)}^{1/2}}\exp ( - \frac{{K{\pi ^2}{{\hat n}^2}}}{\beta })} \nonumber\\
&\sim {\left( {\frac{K}{\beta }} \right)^{1/2}}{\vartheta _3}(0,\exp ( - \frac{{{\pi ^2}K}}{\beta })) \sim {\left( {\frac{\beta }{K}} \right)^{1/2}}{\vartheta _3}(0,\exp ( - \frac{\beta }{K})) \nonumber\\
&{Z_{{\sigma _1}}}(\beta |\phi  = 0) = \sum\limits_{\hat n} {{{( - 1)}^{\hat n}}{Z_{\hat n}}}  \sim {\left( {\frac{K}{\beta }} \right)^{1/2}}\sum\limits_{\hat n} {{{( - 1)}^{\hat n}}\exp ( - \frac{{K{\pi ^2}{{\hat n}^2}}}{\beta })}  \nonumber\\
&\sim {\left( {\frac{K}{\beta }} \right)^{1/2}}{\vartheta _4}(0,\exp ( - \frac{{{\pi ^2}K}}{\beta })) \sim {\left( {\frac{\beta }{K}} \right)^{1/2}}{\vartheta _4}(0,\exp ( - \frac{\beta }{K}))
\end{align}
where $\vartheta_a(u,q)$ is the elliptic theta function 
\begin{align}
&{\vartheta _3}(u,q) = \sum\limits_n {{q^{{n^2}}}\exp (2ni\pi z)}  \equiv \sum\limits_n {{q^{{n^2}}}{\eta ^n}}=\sum\limits_n {\exp ({n^2}\pi i\tau )\exp (2niz)}  \nonumber\\
&{\vartheta _4}(u,q) = \sum\limits_n {{{( - 1)}^n}{q^2}\exp (2ni\pi z)}  = \sum\limits_n {{{( - 1)}^n}{q^{{n^2}}}{\eta ^n}} \nonumber\\
&{\vartheta _2}(u,q) = \sum\limits_n {{q^{{{(n + 1/2)}^2}}}\exp ((2n + 1)i\pi z)}  = \sum\limits_n {{q^{{{(n + 1/2)}^2}}}{\eta ^{n + 1/2}}} 
\end{align}
And we have used the Jacobi identity for elliptic theta function to obtain the final formula of those expressions
\begin{align}
{\vartheta _3}(\frac{z}{\tau }, - \frac{1}{\tau }) = {( - i\tau )^{1/2}}\exp (\frac{\pi }{\tau }i{z^2}){\vartheta _3}(z,\tau )\nonumber\\
{\vartheta _2}(\frac{z}{\tau }, - \frac{1}{\tau }) = {( - i\tau )^{1/2}}\exp (\frac{\pi }{\tau }i{z^2}){\vartheta _4}(z,\tau )
\end{align}
For the whole partition function with the chemical potential $\phi$, we have \cite{Stanford:2017thb}
\begin{align}
{Z_{\hat n}}(\beta ,\phi ) = {Z_{\hat n + \phi /2\pi }}(\beta ,\phi  = 0)
\end{align}
Thus we obtain
\begin{align}
&{Z_{{\sigma _0}}}(\beta |\phi ) = \sum\limits_{\hat n} {{Z_{\hat n + \phi /2\pi }}(\beta ,\phi  = 0)} \nonumber\\
&\sim {\left( {\frac{K}{\beta }} \right)^{1/2}}\sum\limits_{\hat n} {\exp ( - \frac{{{\pi ^2}K{{(\hat n + \phi /2\pi )}^2}}}{\beta })}  \sim {\left( {\frac{\beta }{K}} \right)^{1/2}}{\vartheta _3}(\frac{\phi }{2},\exp ( - \frac{\beta }{K}))\nonumber\\
&{Z_{{\sigma _1}}}(\beta |\phi ) = \sum\limits_{\hat n} {{{( - 1)}^{\hat n}}{Z_{\hat n + \phi /2\pi }}(\beta ,\phi  = 0)} \nonumber\\
&= {\left( {\frac{K}{\beta }} \right)^{1/2}}\sum\limits_{\hat n} {{{( - 1)}^{\hat n}}\exp ( - \frac{{{\pi ^2}K{{(\hat n + \phi /2\pi )}^2}}}{\beta })}  = {\left( {\frac{\beta }{K}} \right)^{1/2}}{\vartheta _2}(\frac{\phi }{2},\exp ( - \frac{\beta }{K}))
\end{align}
Now we apply the Poisson resummation formula to obtain the partition function in the single charge sector. It is easily shown that
\begin{align}
&\sum\limits_{\hat n} {\int_0^{2\pi } {\frac{{d\alpha }}{{2\pi }}\exp ( - i\alpha m)} f(\hat n + \frac{\alpha }{{2\pi }})}  = \int_\mathbb{R} {\exp ( - 2\pi imu)f(u)} du\text{ for integer }m\nonumber\\
&\sum\limits_{\hat n} {\int_0^{2\pi } {\frac{{d\alpha }}{{2\pi }}\exp ( - i\alpha m)} f(\hat n + \frac{\alpha }{{2\pi }})} {( - 1)^n}{\rm{ }} = \int_\mathbb{R} {\exp ( - 2\pi imu)f(u)du} \text{ for half integer }m
\end{align}
Thus, the Fourier transformation formula gives
\begin{align}
Z(\beta ,\mu )\sim  \exp \left( { - \frac{{{\mu ^2}\beta }}{{K}}} \right)
\end{align}
The above computation shows a toy example about partition functions in various cases. Now we could make some simple analysis on those results. 

For $K\gg \beta$, firstly, for the single charge sector results, we will see that for $\mu  \ll {\left( {\frac{K}{\beta }} \right)^{1/2}}$, the partition function is nearly
\begin{align}
Z(\beta ,\mu )\sim 1
\end{align}
while for $\mu  \sim {\left( {\frac{K}{\beta }} \right)^{1/2}}$ or even larger, the result will start to get exponential decaying when $\beta$ increases as 
\begin{align}
Z(\beta ,\mu )\sim\exp \left( { - \frac{{{\mu ^2}\beta }}{{K}}} \right)
\end{align}
Note that there is no leading polynomial dependence on $\beta$.

Secondly, for the whole charge sector with a chemical potential, we have two cases. Firstly, if $\phi/2\pi=n_\phi$ is an integer, the dominated result is simply given by 
\begin{align}
Z(\beta|\phi)\sim{\left( {\frac{K}{\beta }} \right)^{1/2}}
\end{align}
for both spin structures. If $\phi/2\pi$ is not an integer, we write $\text{Round}(x)$ as integer closest to $x$, then we have
\begin{align}
{Z_\sigma }(\beta |\phi )\sim {\left( {\frac{K}{\beta }} \right)^{1/2}}{e^{ - \frac{K}{{4\beta }}{{(\phi  - 2\pi {\rm{Round}}\left( {\frac{\phi }{{2\pi }}} \right))}^2}}}
\end{align}
With similar but more technical analysis, we will generalize the above computations in a general semisimple compact group $G$.

We also notice that for $K \ll \beta$, both canonical and grand canonical ensemble results give constant contribution $\mathcal{O}(1)$ \cite{pcds}. Namely, we could not observe any features from global symmetry sectors. Going back to SYK-like models, we will recover the Schwarzian theory. Thus, this indicates that in the single charge sector one could obtain random matrix theory classification. In case of complex SYK model, it is worked out in \cite{You:2016ldz} by level statistics, and in $\mathcal{N}=2$ supersymmetric SYK model the classification is addressed in \cite{Kanazawa:2017dpd}.

\subsection{A generalized Peter-Weyl theorem}
Now we will study sigma models on the group manifold $G$\footnote{For related mathematics, see \cite{Knapp:LieGroup,Humphreys:LieAlg}.}. The sigma model on a group manifold with a fixed spin structure is described by the Lagrangian
\begin{align}
S_f[f]=-\frac{K\pi}{2}\int _0^{2\pi}\text{Tr}(f^{-1}\partial _t f)^2dt
\end{align}
with respect to the boundary condition that $\tilde f(2\pi)=\tilde f(0)g$. Here $f$ is a group element of $G$, and $\tilde f$ is the lift of $f$ from $G$ to universal cover $\widetilde G$, and $g$ is a central element in $\widetilde G$ such that $\sigma(g)=1$ in $\mathbb Z_2$. 

More precisely, $g$ lives in the kernel of $\widetilde G\to G$, which is a discrete normal subgroup. We claim that
\begin{thm}
\label{thm:small}
$g$ lives in the center of $\widetilde G$.
\end{thm}
\begin{proof}
Every element of form $hgh^{-1}$ is in the kernel of $\widetilde G\to G$, connect $h$ with the identity element of $\widetilde G$ via a path $h(t)$ with $h(0)=\mathbf{1}_{\widetilde G}$ and $h(1)=h$, then the path $h(t)gh(t)^{-1}$ connects $g$ and $hgh^{-1}$, but the kernel of $\widetilde G\to G$ is finite, hence $g$ and $hgh^{-1}$ are equal, i.e. $g$ is central. 
\end{proof}

Furthermore, $g$ can be identified with an element of the fundamental group of $G$ via connecting the identity of $\widetilde G$ with $g$ by a path and projecting down to $G$, the projection of that path is a loop because the head and tail are mapped to the same point (identity of $G$). 

We alos note that preimage of central element in $G$ is still central, in fact since $Z(G)$ is normal dicrete so is its preimage. Conversely the image of central element in $\widetilde G$ is obviously central. Thus there is a surjective homomorphism $Z(\widetilde G)\to Z(G)$, with the same kernel as $\widetilde G\to G$, which is naturally identified with the fundamental group $\pi_1(G)$
\begin{align}\label{center=pi_1}
\pi_1(G)\cong\text{Ker} (Z(\widetilde G)\to Z(G))
\end{align}

The partition function $\text{tr}( e^{-\beta H})$ is the same as the propagator of quantum mechanics on $G$ with Hamiltonian $\widetilde H=-\Delta /2K\pi$, moving from identity element of $G$ to $g$, with duration $2\pi \beta$,
\begin{align}
Z_{\sigma}(\beta)=\langle e^{-\beta H}\rangle=\int_{\tilde f(0) = \tilde f(2\pi )} {[Df]{e^{ - 2\pi \beta \tilde H(f)}}} =\langle \mathbf{1}_G|e^{-2\pi\beta \widetilde H}|\mathbf{1}_G\rangle
\end{align}
where $\Delta$ is the Laplace-Beltrami operator on $G$ associated to the Killing metric $h_{\mu \nu}$\footnote{Killing metric is defined at the tangent space of identity to be $\langle X,Y\rangle=\text{Tr}(\text{ad}(X)\text{ad}(Y))$, then pushforward to the tangent space at each element $g$ by left multiplication $L_g$ (or equivalently right multiplication $R_g$, because Killing metric at identity is invariant under adjoint action).}, defined in the usual way
\begin{align}\label{Laplacian}
\Delta f=\frac {1}{\sqrt{\text {det}(h)}}\partial _{\mu}\left(\sqrt{\text {det}(h)} \partial ^{\mu}f\right)
\end{align}
Note that the Laplace-Beltrami operator acts on the bundle of \textit{twisted} functions, i.e. the function $f$ in \ref{Laplacian} should be a local section of the complex line bundle $\mathcal L_{\sigma}$ coming from the spin structure $\sigma\in \text{H}^1(G,\mathbb Z_2)$\footnote{More precisely, an element $\sigma$ in $\text{H}^1(G,\mathbb Z_2)\cong \text{Map}(G,\text{B}\mathbb Z_2)$ determines a real line bundle, and tensoring with $\mathbb C$ gives rise to a complex line bundle. Equivalently, the representative of that complex line bundle in $\text{H}^2(G,\mathbb Z)$ is the image of $\sigma$ under the Bockstein homomorphism.}. The propagator is calculated by decomposing into eigenfunctions of $\widetilde H$, 
\begin{align}\label{Expansion}
\langle \mathbf{1}_G|e^{-2\pi\beta \widetilde H}|\mathbf{1}_G\rangle&=\sum _{i}\psi_i(\mathbf{1}_G)\bar {\psi}_n(\mathbf{1}_G)e^{-2\pi\beta E_i}
\end{align}
where $\psi_i$ is the eigenfunction of $\widetilde H$ with eigenvalue $E_i$. To give a description of these eigenfunctions, let's first assume that the spin structure is trivial so that the line bundle $\mathcal L_{\sigma}$ is trivial and functions are ordinary, i.e. not twisted. Recall the famous Peter-Weyl theorem:
\begin{thm}
\label{thm:PeterWeyl}
    $\bf{\operatorname{[Peter-Weyl]}}$ Let $G$ be a compact Lie group equipped with the Haar measure, then the Hilbert space of square-integrable functions on $G$ is a unitary representation of $G$ by the action $\pi(g)$
    \begin{align*}
        \pi(g):f(h)\mapsto f(g^{-1}h)
    \end{align*}
and has decomposition into finite dimensional irreducible representations:
    \begin{align}\label{PeterWeyl}
        L^2(G)=\bigoplus _{\lambda\in P(G)\cap P_+}V_{\lambda}^{\oplus \operatorname{dim}(V_{\lambda})}
    \end{align}
    Here $V_{\lambda}$ is the unitary irreducible representation of highest weight $\lambda$, and $P(G)$ is the weight lattice of $G$ \footnote{Weight lattice $P(G)$ is a lattice that labels all possible weights in the representations of $G$.} and $P_+$ is the dominant part of weight space $P(G)\otimes _{\mathbb Z} \mathbb R$ \footnote{Tensor over $\mathbb Z$ means forming a tensor product $\mathbb Z$-bilinearly, here $P(G)\otimes _{\mathbb Z} \mathbb R$ embeds the $P(G)$ lattice into a real linear space whose the dimension equals to the rank of the lattice. Dominant part $P_+$ is the domain in the weight space such that $\forall \lambda \in P_+, \langle \lambda,\alpha_i\rangle\ge 0$, where $\alpha_i$ runs through all positive roots. Dominant weights $P(G)\cap P_+$ are one to one correspond to unitary irreducible representations of $G$.}. This isomorphism is given by taking the matrix coefficients of each irreducible representation, more precisely, let $\lambda$ be a weight, and $\pi ^{\lambda}:G\to \text{U}(V_{\lambda})$ be the associated unitary irreducible representation with highest weight $\lambda$, $\{e_i\}$ be an orthonormal basis of $V_{\lambda}$ with Hermitian metric $(-,-)$, then following functions on $G$
    \begin{align}
        \pi ^{\lambda}_{ij}(g):=\sqrt{\operatorname{dim}(V_{\lambda})}(\pi ^{\lambda}(g)e_i,e_j)
    \end{align}
constitute an orthonormal basis for the direct summand $V_{\lambda}^{\oplus \operatorname{dim}(V_{\lambda})}$ in the decomposition \ref{PeterWeyl}.
\end{thm}
Consider a left invariant vector fields $X$ acting on $L^2(G)$, its Lie derivative on a function $f$ is by definition the infinitesimal generator of Lie group action on function $f$, thus it agrees with the action of Lie algebra element $X(\mathbf{1}_G)$
\begin{align*}
    \mathcal L_X f=\pi(X(\mathbf{1_G}))f
\end{align*}
and by associativity of the Lie algebra action, every differential operator $D$
which is constructed from left invariant vector fields ({where $n$ is the dimension of the group})
\begin{align*}
D=X_1X_2\cdots X_n
\end{align*}
acts on functions by
\begin{align*}
Df=\pi(X_1 (\mathbf{1}_G) \otimes X_2 (\mathbf{1}_G)\otimes\cdots X_n (\mathbf{1}_G))f
\end{align*}
here $X_1 (\mathbf{1}_G) \otimes X_2 (\mathbf{1}_G)\otimes\cdots X_n (\mathbf{1}_G) $ is regarded as an element in the universal enveloping algebra $\mathcal U(\mathfrak g)$. As a corollary, the Laplace-Beltrami operator $\Delta$, which equals to 
\begin{align*}
    \sum_i X_iX_i
\end{align*}
where $\{X_1,X_2,\cdots, X_n\}$ is an orthonormal basis (under Killing metric) for left invariant vector fields, acts on functions by 
\begin{align*}
    \Delta f=\sum _i \pi(X_i (\mathbf{1}_G) \otimes X_i (\mathbf{1}_G))f
\end{align*}
but $\{X_1 (\mathbf{1}_G),X_2 (\mathbf{1}_G),\cdots, X_n (\mathbf{1}_G)\}$ is an orthonormal basis (under Killing metric) for the Lie algebra $\mathfrak g$, so $\sum _i X_i (\mathbf{1}_G) \otimes X_i (\mathbf{1}_G) $ is the second order Casimir operator in $\mathcal U(\mathfrak g)$, and it acts on irreducible representation $V_{\lambda}$ by a scalar $C_2(\lambda)$ which equals to 
\begin{align}
\langle\lambda,\lambda+2\rho\rangle
\end{align}
where $\rho$ is one half of the sum of positive roots, and the inner product is the one induced from the Killing metric.

Now eigenfunction $\psi _i(g)$ in the expansion formula \ref{Expansion} is an one-to-one correspondence to $\pi ^{\lambda}_{ij}$, thus the \ref{Expansion} reads
\begin{align}\label{Expansion2_No_Spin}
    Z_{\sigma=0}(\beta)&=\sum_{\lambda \in P(G) \cap P_+}\sum_{i,j=1}^{\text{dim}(V_{\lambda})}\pi^{\lambda}_{ij}(\mathbf 1_G)\bar{\pi}^{\lambda}_{ij}(\mathbf 1_G) e^{-\beta C_2(\lambda)/K}\nonumber\\
    &= \sum_{\lambda \in P(G) \cap P_+}\sum_{i,j=1}^{\text{dim}(V_{\lambda})}\text{dim}(V_{\lambda})|(\pi^{\lambda}(\mathbf 1_G)e_i,e_j)|^2e^{-\beta C_2(\lambda)/K}
\end{align}
Note that $\pi^{\lambda}(\mathbf 1_G) $ is nothing but identity matrix in vector space $V_{\lambda}$, hence 
\begin{align}
    (\pi^{\lambda}(\mathbf 1_G)e_i,e_j)=(e_i,e_j)=\delta_{ij}
\end{align}
and \ref{Expansion2_No_Spin} reduces to 
\begin{align}\label{Expansion3_No_Spin}
Z_{\sigma=0}(\beta) =\sum_{\lambda\in P(G)\cap P_+}\text{dim}(V_{\lambda})^2e^{-\beta C_2(\lambda)/K}
\end{align}

Now we need to remove the assumption on the triviality of spin structure, i.e. the Hilbert space of the quantum mechanics on $G$ should be square integrable global sections of a nontrivial complex line bundle $\mathcal L_{\sigma}$. One expects that there should be a decomposition of the Hilbert space into direct sum of irreducible representations of $G$, similiar to the Theorem \ref{thm:PeterWeyl}. However this is not the case because there is no self-consistent action of $G$ on the Hilbert space $L^2(G,\mathcal L_{\sigma})$ such that it's compatible with the translation $L_{g^{-1}}:h\mapsto g^{-1}h$, i.e. $\mathcal L_{\sigma}$ is not a $G$-equivalent line bundle. Assume that there is a action of $G$ on $L^2(G,\mathcal L_{\sigma})$ compatible with translation, then it's represented by an isomorphism
\begin{align*}
a_g:L_{g^{-1}}^*\mathcal L_{\sigma}\cong \mathcal L_{\sigma}
\end{align*}
let $g$ run through the whole group and it amounts to an isomorphism between line bundles on $G\times G$
\begin{align}
a:m^*\mathcal L_{\sigma}\cong p_2^*\mathcal L_{\sigma}
\end{align}
in which $m$ is the multiplication map $(g_1,g_2)\mapsto g_1g_2$ and $p_2$ is the projection to the second coordinate $(g_1,g_2)\mapsto g_2$. Restricting to $G\times \{\mathbf{1}_G\}\subset G\times G$, there is an isomorphism 
\begin{align}
a|_{G\times \{\mathbf{1}_G\}}:\text{Id}_G^*\mathcal L_{\sigma}\cong \mathbf{1}_G^*\mathcal L_{\sigma}
\end{align}
$\mathbf{1}_G$ means collapsing $G$ to a point followed by embedding into the identity element $\mathbf{1}_G$. This is an isomorphism between $\mathcal L_{\sigma}$ and trivial bundle, a contradiction to the fact that $\mathcal L_{\sigma}$ is nontrivial.

This drawback is rescued by considering the $\widetilde G$-equivalent structure of $\mathcal L_{\sigma}$. In fact, the line bundle $\mathcal L_{\sigma}$ carries a canonical flat connection which comes from the construction: $\sigma\in \text{Hom}(\pi _1(G),\mathbb Z_2)$ determines a $\mathbb Z_2$-principal bundle which is obviously flat (there is no vertical direction), $\mathbb Z_2$'s action on $\mathbb C$ gives rise to a associated complex line bundle with a flat connection inherited from the $\mathbb Z_2$-principal bundle. Now we can define the action of $\widetilde G$ on $\mathcal L_{\sigma}$ by connecting an element $g\in\widetilde G$ with the identity element $\mathbf{1}_{\widetilde G}$ via a smooth path $g(t)$, and let the horizontal lift of the left multiplication $L_{g(t)}$ be the action of $g$, this does not depend on the choice of path because $\widetilde G$ is simply-connected and the connection is flat. This is obviously a group action because composition of any two elements $g_1$ and $g_2$ amounts to gluing path from $\mathbf{1}_{\widetilde G}$ to $g_2$ and path from $g_2$ to $g_1g_2$, which is a path from $\mathbf{1}_{\widetilde G}$ to $g_1g_2$.

Pull-back of $\mathcal L_{\sigma}$ to $\widetilde G$ is the trivial line bundle whose square integrable global sections have decomposition into irreducible representations
\begin{align}
L^2(\widetilde G)=\bigoplus _{\lambda\in P(\widetilde G)\cap P_+}V_{\lambda}^{\oplus \text{dim}(V_{\lambda})}
\end{align}
A question to be answered is: which $L^2(\widetilde G)$ function comes from a global section of $\mathcal L_{\sigma}$ on $G$? A necessary condition is that $f(g^{-1}h)=\sigma (g)f(h)$, $\forall h\in \widetilde G$ and $\forall g\in \text{Ker}(\widetilde G\to G)$. This comes from the monodromy that of any loop $\gamma(t)$ in $G$ is $\sigma ([\gamma(t)])\in \mathbb Z_2$, and $g$ is canonically identified with an element in $\pi_1(G)$ by connecting it with $\mathbf {1}_{\widetilde G}$ via a path and the monodromy generated the image of this path (which is a loop) is by definition the action of $g$ on the section. It turns out that this is also sufficient. Let $ \text{Ker}(\widetilde G\to G)$ acts on the trivial bundle via $g(h,u)=(gh,\sigma (g)u)$, $\forall h\in \widetilde G$ and $u\in \mathbb C$, then this action is compatible with left multiplication of $ \text{Ker}(\widetilde G\to G)$ on $\widetilde G$, and it also preserve the trivial connection, thus the trivial bundle descends to a line bundle $\mathcal L_{\sigma}'$ on $G$ with a flat connection. Monodromy of $\mathcal L_{\sigma}'$ is exactly $\sigma $, so $\mathcal L_{\sigma}'$ and $\mathcal L_{\sigma}$ have the same monodromy, indicating that they are isomorphic, since line bundles on $G$ are classified by $\text{Map}(G,\text{BU}(1))\cong \text{Hom}(\pi_1(G),\text{U}(1))$, i.e. monodromy. 

We know that $\text{Ker}(\widetilde G\to G)$ is a subgroup of the center of $\widetilde G$, so their action on irreducible representation $V_{\lambda}$ are scalars (Schur's Lemma),  it remains to pick out those $\lambda$'s such that these scalars are exactly $\sigma(g)$. It's attempting to extend the definition of $\sigma$ and let it act on the whole group $\widetilde G$ so that it corresponds to a weight, it turns out that this is possible, modulo weight lattice $P(G)$:
\begin{lem}
There is a canonical isomorphism
\begin{align}
\operatorname{Hom}(\pi_1(G),\mathbb Z_2)\cong (P(\mathfrak g)\cap \frac {1}{2} P(G))/P(G)
\end{align}
so any spin structure lifts to a weight in $(P(\mathfrak g)\cap \frac {1}{2} P(G)$, also denoted by $\sigma$, defined up to $P(G)$, such that its action on $\pi_1(G)$ is $\sigma$. 
\end{lem}
\begin{proof}
In fact there are isomorphisms
\begin{align}
\text{Hom}(\pi_1(G),\mathbb Z_2)&\cong \text{Hom}(\text{Ker}(Z(\widetilde G)\to Z(G)),\mathbb Z_2)\nonumber\\
&\cong\text{Hom}((P(\mathfrak g)/P(G))^*,\mathbb Z_2)\nonumber\\
&\cong (P(\mathfrak g)\cap \frac {1}{2} P(G))/P(G)
\end{align}
Here the dual group $(P(\mathfrak g)/P(G))^*$ is defined as the Pontryagin dual $\text{Hom}(P(\mathfrak g)/P(G),\mathbb Q/\mathbb Z)$. The first isomorphism comes from \ref{center=pi_1}, the second isomorphism can be proved as following: It's well-known that maximal tori are conjugated with each other and they cover the whole group \cite{Knapp:LieGroup} , in particular every element in the center of $\widetilde G$ lies in the intersection of maximal tori (because it belongs to at least one maximal torus, then adjoint action take this particular maximal torus to other maximal tori). Now pick one maximal torus $T$, then $Z(\widetilde G)\subset T$, and $\forall x\in T$, $x$ can written as $x=e^{2\pi iX}$ for $X\in \mathfrak g$, note that $X$ is defined modulo $P(\widetilde G)^*$  since $\forall \lambda \in P(\widetilde G)$ and $\forall t\in P(\widetilde G)^*$, we have $e^{2\pi i\langle t,\lambda\rangle}=1$. $x\in Z(\widetilde G)$ if and only if the adjoint action $\text{Ad}(e^{2\pi iX})$ is trivial, or equivalently $\forall \alpha$ in the root system of $\mathfrak g$, $\langle \alpha,X\rangle\in \mathbb Z$, hence we can identify $Z(\widetilde G)$ with $Q^*/P(\widetilde G)^*$, where $Q$ denotes the root lattice and $Q^*$ is its dual. The same argument applies to $Z(G)$, and there is identification $Z(G)\cong Q^*/P(G)^*$. Now 
\begin{align}
\text{Ker}(Z(\widetilde G)\to Z(G))\cong \text{Ker}(Q^*/P(\widetilde G)^*\to Q^*/P(G)^*)\cong (P(\widetilde G)/P(G))^*
\end{align}
Obviously $P(\widetilde G)\subset P(\mathfrak g)$, and since $\widetilde G$ is simply conncted every Lie algebra representation $\mathfrak g\to \mathfrak {gl}_n$ gives rise to a Lie group representation $\widetilde G\to \text{GL}_n$, thus $P(\widetilde G)= P(\mathfrak g)$, and we arrives at $\text{Ker}(Z(\widetilde G)\to Z(G))\cong (P(\mathfrak g)/P(G))^*$, which implies the second isomorphism. For the third isomorphism, notice that there is a short exact sequence
\begin{center}
\begin{tikzcd}
0\arrow[r]&\mathbb Z_2\arrow[r] &\mathbb Q/\mathbb Z \arrow[r, "2"]& \mathbb Q/\mathbb Z \arrow[r]&0
\end{tikzcd}
\end{center}
which implies that $\text{Hom}((P(\mathfrak g)/P(G))^*,\mathbb Z_2)$ is the kernel of multiplication by 2 on the group $\text{Hom}((P(\mathfrak g)/P(G))^*,\mathbb Q/\mathbb Z)$, and the latter is identified with $P(\mathfrak g)/P(G)$ by Pontryagin duality. The kernel of multiplication by 2 is calculated by elementary group theory to be $(P(\mathfrak g)\cap \frac {1}{2} P(G))/P(G)$.
\end{proof}

Thus we have established the following generalization of Peter-Weyl theorem \ref{thm:PeterWeyl}:
\begin{thm}
\label{thm:Gen_PeterWeyl}
    The same notation as above, then square integrable twisted sections of the line bundle $\mathcal L_{\sigma}$ has a decomposition into finite dimensional irreducible representations of $\widetilde G$:
    \begin{align}\label{Gen_PeterWeyl}
        L^2(G,\mathcal L_{\sigma})=\bigoplus _{\lambda\in (\sigma +P(G))\cap P_+}V_{\lambda}^{\oplus \operatorname{dim}(V_{\lambda})}
    \end{align}
This isomorphism is given by taking the matrix coefficients of each irreducible representation, more precisely, let $\lambda$ be a weight, and $\pi ^{\lambda}:G\to \text{U}(V_{\lambda})$ be the associated unitary irreducible representation with highest weight $\lambda$, $\{e_i\}$ be an orthonormal basis of $V_{\lambda}$ with Hermitian metric $(-,-)$, then following twisted functions on $G$
    \begin{align}
        \pi ^{\lambda}_{ij}(g):=\sqrt{\operatorname{dim}(V_{\lambda})}(\pi ^{\lambda}(\widetilde g)e_i,e_j)
    \end{align}
constitute an orthonormal basis for the direct summand $V_{\lambda}^{\oplus \operatorname{dim}(V_{\lambda})}$ in the decomposition \ref{Gen_PeterWeyl}, where $\widetilde g$ is a lift of $g$ to $\widetilde G$.
\end{thm}
Accordingly, the summation of \ref{Expansion3_No_Spin} should be replaced by dominant weights in the lattice $\sigma +P(G)$ and one arrives at 
\begin{align}\label{Expansion3}
Z_{\sigma}(\beta) =\sum_{\lambda\in (\sigma +P(G))\cap P_+}\text{dim}(V_{\lambda})^2e^{-\beta C_2(\lambda)/K}
\end{align}
A more explicit form of \ref{Expansion3} can be deduced from the Weyl dimension formula. Recall that the Weyl dimension formula 
\begin{align}\label{Weyl_Dim}
\text{dim}(V_{\lambda})=\prod_{\alpha \in R_+}\frac{\langle \alpha,\lambda+\rho\rangle }{\langle \alpha , \rho\rangle }
\end{align}
relates the dimension of an irreducible unitary representation $V_{\lambda}$ with the highest weight $\lambda$ and positive roots $\alpha\in R_+$. Plug it into \ref{Expansion3_No_Spin} and one arrives at
\begin{align}\label{Total}
{Z_{\sigma}(\beta)=\sum_{\lambda\in (\sigma +P(G))\cap P_+}\prod_{\alpha \in R_+}\frac{\langle \alpha,\lambda+\rho\rangle ^2}{\langle \alpha , \rho\rangle ^2}e^{-\beta \langle\lambda,\lambda+2\rho\rangle/K}}
\end{align}
This corresponds to the whole partition function with trivial chemical potential.

\subsection{The single charge sector}
One can also get the single charge sector contribution by applying the chemical potential trick to the partition function
\begin{align}
Z_{\sigma}(\beta|\phi):=\langle e^{-\beta H+i\phi Q}\rangle=\langle e^{i\phi}|e^{-2\pi\beta \widetilde H}|\mathbf{1}_G\rangle
\end{align}
where $\phi$ is an element in Lie algebra $\mathfrak{g}$. Note that this is related to charge sectors by
\begin{align}\label{defn_chargesector}
Z_{\sigma}(\beta|\phi):=\sum_{\mu\in P(G)}Z_{\sigma}(\beta,\mu)e^{i\langle \mu ,\phi\rangle}
\end{align}
Similiar to the last section, one expands the Hamiltonian $\widetilde H$ with respect to its eigenfunctions which have been fully classified by the Peter-Weyl theorem \ref{Gen_PeterWeyl}, and concludes that 
\begin{align}\label{Twist_1}
Z_{\sigma}(\beta|\phi)&=\sum _{i}\psi_i(e^{i\phi})\bar {\psi}_i(\mathbf{1}_G)e^{-2\pi\beta E_i}\nonumber\\
&=\sum_{\lambda \in (\sigma +P(G)) \cap P_+}\sum_{i,j=1}^{\text{dim}(V_{\lambda})}\pi^{\lambda}_{ij}(e^{i\phi})\bar{\pi}^{\lambda}_{ij}(\mathbf 1_G) e^{-\beta C_2(\lambda)/K}
\nonumber\\
&= \sum_{\lambda \in (\sigma +P(G)) \cap P_+}\sum_{i,j=1}^{\text{dim}(V_{\lambda})}\text{dim}(V_{\lambda})(\pi^{\lambda}(e^{i\phi})e_i,e_j)\overline{(\pi^{\lambda}(\mathbf 1_G)e_i,e_j)}e^{-\beta C_2(\lambda)/K}
\end{align}
$\pi^{\lambda}(\mathbf 1_G)$ is just the identity matrix on $V_{\lambda}$, which gives a $\delta_{ij}$ in the summation and turns $(\pi^{\lambda}(e^{i\phi})e_i,e_j)$ into a trace $\text{Tr}\left(\pi^{\lambda}(e^{i\phi})\right)$, and by definition this is the character $\chi _{\lambda}(\phi)$ of representation $V_{\lambda}$ thus the partition function with chemical potential reads
\begin{align}\label{Twist_2}
Z_{\sigma}(\beta|\phi)= \sum_{\lambda \in (\sigma +P(G)) \cap P_+}\text{dim}(V_{\lambda})\chi _{\lambda}(\phi)e^{-\beta C_2(\lambda)/K}
\end{align}

On the other hand, the basis of $V_{\lambda}$ can be chosen to be weight vectors such that the action of $e^{i\phi}$ is through $e^{i\langle \mu ,\phi\rangle}$ for each weight $\mu$ so the character $\chi _{\lambda}(\phi)$ can be represented by weight space decomposition 
\begin{align}\label{char}
\chi _{\lambda}(\phi)=\sum_{\mu\text{ is a weight in }V_{\lambda}}e^{i\langle \mu ,\phi\rangle}
\end{align}
Bring \ref{Twist_2} and \ref{char} together and plug them into the definition of charge sectors \ref{defn_chargesector}, one can write down the partition function for a single charge $\mu$
\begin{align}
Z_{\sigma}(\beta,\mu)=\sum_{\lambda\in (\sigma +P(G))\cap P_+}\text{dim}(V_{\lambda})\text{dim}_{\mu}(V_{\lambda})e^{-\beta C_2(\lambda)/K}
\end{align}
$\text{dim}_{\mu}(V_{\lambda})$ is the dimension of subspace of $V_{\lambda}$ with weight $\mu$. Weyl's dimension formula produces an explicit form of $\text{dim}(V_{\lambda})$, and recall the Kostant's dimension formula
\begin{align}
\text{dim}_{\mu}(V_{\lambda})=\sum_{w\in \mathcal W}(-1)^{|w|}\mathcal P(w(\lambda +\rho)-(\mu+\rho))
\end{align}
$\mathcal W$ is the Weyl group acting on weights, $|w|$ is the length of the elemnt $w$, i.e. the smallest number of $\alpha$'s such that $w$ can be generated as multiplication of reflections $w=s_{\alpha_1}\cdots s_{\alpha_{|w|}}$ \footnote{$s_{\alpha}$ is the reflection of the root plane with respect to the axis $\alpha$.}\cite{Humphreys:LieAlg}, $\mathcal P$ is the function that for each $\nu \in P(g)$, $\mathcal P(\nu)$ is the number of nonnegative integer solution $\{k_{\alpha}\}_{\alpha\in R_+}$ to the equation
\begin{align}
\nu=\sum_{\alpha\in R_+}k_{\alpha}\alpha
\end{align}
With the help of these fomulas, $Z_{\sigma}(\beta,\mu)$ can be more explicit
\begin{align}\label{ChargeSector_1}
Z_{\sigma}(\beta,\mu)=\sum_{\lambda\in (\sigma +P(G))\cap P_+}\prod_{\alpha \in R_+}\frac{\langle \alpha,\lambda+\rho\rangle }{\langle \alpha , \rho\rangle }\sum_{w\in \mathcal W}(-1)^{|w|}\mathcal P(w(\lambda +\rho)-(\mu+\rho))e^{-\beta C_2(\lambda)/K}
\end{align}
Since $\forall w\in \mathcal W$, $\text{dim}_{\mu}(V_{\lambda})=\text{dim}_{w(\mu)}(V_{\lambda})$, one can always conjugate $\mu $ to a dominant one without changing the partition function, so the assumption that $\mu$ is dominant can be made. Another inspection is that $w$ can written as reflections $w=s_{\alpha_1}\cdots s_{\alpha_{|w|}}$, $s_{\alpha_1}$ turns $\alpha_1$ into $-\alpha_1$ and permutes other $\alpha_i$'s so 
\begin{align}
\prod_{\alpha \in R_+}\langle \alpha,s_{\alpha_1}\cdots s_{\alpha_{|w|}}(\lambda +\rho)\rangle& =-\prod_{\alpha \in R_+}\langle \alpha,s_{\alpha_2}\cdots s_{\alpha_{|w|}}(\lambda +\rho)\rangle \nonumber\\
&=\cdots =(-1)^{|w|}\prod_{\alpha \in R_+}\langle \alpha,\lambda +\rho\rangle 
\end{align}
Plug this equation into \ref{ChargeSector_1} and simplifies it into
\begin{align}\label{ChargeSector_2}
Z_{\sigma}(\beta,\mu)&=\sum_{\lambda\in (\sigma +P(G))\cap P_+}\sum_{w\in \mathcal W}\prod_{\alpha \in R_+}\frac{\langle \alpha,w(\lambda +\rho)\rangle }{\langle \alpha , \rho\rangle }\mathcal P(w(\lambda +\rho)-(\mu+\rho))e^{-\beta C_2(\lambda)/K}\nonumber\\
\end{align}
Since $\mathcal P$ function counts nonnegative solutions $\{n_{\alpha}\}$ to the equation
\begin{align}
w(\lambda +\rho)=\mu+\rho+\sum_{\alpha\in R_+}n_{\alpha}\alpha=\mu+\rho+\vec n\cdot \vec {\alpha}
\end{align}
it makes no harm to replace $w(\lambda +\rho)$ by $\mu+\rho+\vec n\cdot \vec {\alpha}$. On the other hand one can also replace the summation over $\lambda$ and $w$ by summation over $\vec n\in \mathbb Z_{\ge 0}^p$, in which $p$ is the number of positive roots, i.e. the number of elements in $R_+$. We claim that this is possible, i.e 
\begin{align}\label{ChargeSector_3}
Z_{\sigma}(\beta,\mu)=\sum_{\vec{n}\in \mathbb Z_{\ge 0}^p}\prod_{\alpha \in R_+}\frac{\langle \alpha,\mu+\rho+\vec n\cdot \vec {\alpha}\rangle }{\langle \alpha , \rho\rangle }e^{-\beta (|\mu+\rho+\vec n\cdot \vec {\alpha}|^2-|\rho|^2)/K}
\end{align}
To prove this equation, it suffices to show the equivalence between two index sets, i.e. $\forall \vec n\in \mathbb Z_{\ge 0}^p$, there exists a unique combination $\lambda$ and $w$ such that $w(\lambda +\rho)=\mu+\rho+\vec n\cdot \vec {\alpha}$. This statement is equivalent to that $\mu+\rho+\vec n\cdot \vec {\alpha}$ lies in the interior of some Weyl chamber: it's necessary because $\lambda +\rho$ is inside the interior of the dominant Weyl chamber and Weyl group permutes interior of the Weyl chambers. It's also sufficient because there is a unique Weyl group element $w_0$ sending it to the interior of dominant Weyl chamber, i.e. its coordinates under the Dynkin basis are all positive intergers, so $\lambda_0:=w_0(\mu+\rho+\vec n\cdot \vec {\alpha})-\rho\in P(\mathfrak g)\cap P_+$. The only thing remains to verify is that $\lambda_0 \in \sigma +\frac{1}{2}P(G)$, this is done by observing that 
\begin{align}
&w_0(\vec n\cdot \vec {\alpha})\in Q\subset P(G)\nonumber\\
\forall \beta\in& P(\mathfrak g), w_0(\beta)-\beta\in Q\subset P(G)
\end{align}
so $\lambda_0\in \mu+P(G)$, but $\mu\in \sigma +\frac{1}{2}P(G)$, whence $\lambda_0 \in \sigma +\frac{1}{2}P(G)$. However it's totally possible that $\mu+\rho+\vec n\cdot \vec {\alpha}$ lies in the boundary of some Weyl chamber, take $\text{SU}(3)$ for example, $\alpha_1$ and $\alpha_2$ are two simple roots of it, then $\rho=\alpha_1+\alpha_2$ and $\rho+\alpha_2$ is in the boundary of dominant Weyl chamber, because $\langle \rho+\alpha_2,\alpha_1\rangle=0$. Nevertheless, this does not affect the summation because some $\alpha\in R_+$ kills $\mu+\rho+\vec n\cdot \vec {\alpha}$ and the product term in the summation is automatically zero, i.e. the only survivals are those $\mu+\rho+\vec n\cdot \vec {\alpha}$  lying in the interior of some Weyl chamber, and the equivalence between two index sets is established.
\\
\\
One finally arrives at
\begin{align}\label{ChargeSector_Final}
{Z_\sigma }(\beta ,\mu ) = \sum\limits_{\vec n \in \mathbb{Z}_{ \ge 0}^p} F (\mu  + \rho  + \vec n \cdot \vec \alpha ){\rm{ ,  }}F(\nu ) = {e^{ - \beta (|\nu {|^2} - |\rho {|^2})/K}}\prod\limits_{\alpha  \in {R_ + }} {\frac{{\langle \alpha ,\nu \rangle }}{{\langle \alpha ,\rho \rangle }}} 
\end{align}
In the current stage, we may observe that when $\beta \gg K$, the formula simply gives a constant, ${Z_\sigma }(\beta ,\mu )\sim \mathcal{O}(1)$. In this limit, there is no contribution from the symmetry sector purely from the partition function. However, when $\beta \ll K$, we would still expect some interesting dependence over temperature.

\subsection{Partition function with fixed chemical potential}
We have already derived the total partition function \ref{Total} for the fixed chemical potential $\phi=0$. For generic chemical potential, an easy way is to use the resummation formula given in the appendix of \cite{Picken:1988ev}, based on the single charge sector result. In this section we will show that the partition function with given chemical potential is given by
\begin{align}\label{Twist_3}
Z_{\sigma}(\beta|\phi)=c\left(\frac{K}{4\pi \beta}\right)^{n/2}e^{\frac{n\beta}{24K}}\Theta_{\sigma}(\beta,\phi)
\end{align}
$c$ is a constant depending on $G$, which is computed by
\begin{align}
c=(2\pi)^{p+r}(\text{det }\mathcal C)^{1/2}\prod_{\alpha \in R_+}\langle \alpha,\rho\rangle^{-1}
\end{align}
where $r$ is the rank of $G$, i.e. dimension of Cartan subalgbra, $\mathcal C$ is the Cartan matrix. $\Theta_{\sigma}(\beta,\phi)$ is the theta function defined by
\begin{align}\label{Theta}
\Theta_{\sigma}(\beta,\phi)\equiv \sum_{\mu \in \Lambda(G)}\prod_{\alpha \in R_+}\frac{\langle \alpha, \phi+2\pi \mu\rangle }{2\sin (\langle \alpha, \phi+2\pi \mu\rangle/2 ) }e^{2\pi i\langle \sigma+\rho,  \mu\rangle}e^{-\frac{K}{4\beta}\langle \phi+2\pi \mu,\phi+2\pi \mu\rangle}
\end{align}
In particular, if $\phi$ is taken to be zero, then this gives the partition function
\begin{align}\label{Total_2}
Z_{\sigma}(\beta)&=c\left(\frac{K}{4\pi \beta}\right)^{n/2}e^{\frac{n\beta}{24K}}\Theta_{\sigma}(\beta,0)
\end{align}
First of all, we apply the Weyl character formula and Weyl dimension formula to expand \ref{Twist_2} concretely:
\begin{align}
\begin{split}
Z_{\sigma}(\beta|\phi)&= \sum_{\lambda \in (\sigma +P(G)) \cap P_+}e^{-\beta C_2(\lambda)/K}\left[\prod_{\alpha \in R_+}\frac{\langle \alpha,\lambda+\rho\rangle }{\langle \alpha , \rho\rangle }\right]\\
&\times\left[\prod_{\alpha \in R_+}\frac{1}{2i\sin (\langle \alpha, \phi\rangle/2 )}\sum_{w\in \mathcal W}(-1)^{|w|}\exp (i\langle w(\lambda+\rho),\phi\rangle )\right]
\end{split}
\end{align}
Recall the definition of constant $c=(2\pi)^{p+r}(\text{det }\mathcal C)^{1/2}\prod_{\alpha \in R_+}\langle \alpha,\rho\rangle^{-1}$ and plug it into the formula:
\begin{align}
\begin{split}
Z_{\sigma}(\beta|\phi)&=\frac{c}{(2\pi)^{p+r}(\text{det }\mathcal C)^{1/2}}\sum_{\lambda \in (\sigma +P(G)) \cap P_+}e^{-\beta C_2(\lambda)/K}\\
&\times\left[\prod_{\alpha \in R_+}\frac{\langle \alpha,\lambda+\rho\rangle }{2i\sin (\langle \alpha, \phi\rangle/2 )}\sum_{w\in \mathcal W}(-1)^{|w|}\exp (i\langle w(\lambda+\rho),\phi\rangle )\right]
\end{split}
\end{align}
Since the second Chern number of representation $\lambda$ is $C_2(\lambda)=\langle \lambda+\rho,\lambda+\rho\rangle -\langle \rho,\rho\rangle$, and according to \emph{strange formula} of Freudenthal and de Vries, i.e. $\langle \rho,\rho\rangle=n/24$, where $n=\text{dim}(G)$, we can rewrite the formula as
\begin{align}
\begin{split}
Z_{\sigma}(\beta|\phi)=&\frac{c\cdot \exp (n\beta/24K)}{(2\pi)^{p+r}(\text{det }\mathcal C)^{1/2}}\sum_{\lambda \in (\sigma +P(G)) \cap P_+}\sum_{w\in \mathcal W}\\
&\left[\prod_{\alpha \in R_+}\frac{\langle \alpha,w(\lambda+\rho)\rangle }{2i\sin (\langle \alpha, \phi\rangle/2 )}\right]e^{i\langle w(\lambda+\rho),\phi\rangle-\beta \langle w(\lambda+\rho),w(\lambda+\rho)\rangle /K}
\end{split}
\end{align}
It's obvious that when $\lambda$ runs though all lattice points in
$(\sigma +P(G)) \cap P_+$ and $w$ runs through all group elements in $\mathcal W$, $w(\lambda+\rho)$ runs through all lattice points in $\sigma+\rho +P(G)$ with multiplicity one except possibly for those lying on the boundary of Weyl chambers, but those weights annihilate at least one $\alpha\in R_+$, hence the formula doesn't change if we simply add them to the summation by hand, i.e. we have
\begin{align}
Z_{\sigma}(\beta|\phi)=\frac{c\cdot \exp (n\beta/24K)}{(2\pi)^{p+r}(\text{det }\mathcal C)^{1/2}}\sum_{\lambda \in \sigma+\rho +P(G)}\left[\prod_{\alpha \in R_+}\frac{\langle \alpha,\lambda\rangle }{2i\sin (\langle \alpha, \phi\rangle/2 )}\right]e^{i\langle \lambda,\phi\rangle-\beta \langle \lambda,\lambda\rangle /K}
\end{align}
Using Poisson resummation formula, we can rewrite it as the summation over the dual lattice of $P(G)$, which is be denoted by $\Lambda (G)$:
\begin{align}
Z_{\sigma}(\beta|\phi)=c\left(\frac{K}{4\pi \beta}\right)^{n/2}e^{\frac{n\beta}{24K}}\sum_{\mu \in \Lambda(G)}\prod_{\alpha \in R_+}\frac{\langle \alpha, \phi+2\pi \mu\rangle }{2\sin (\langle \alpha, \phi+2\pi \mu\rangle/2 ) }e^{2\pi i\langle \sigma+\rho,  \mu\rangle}e^{-\frac{K}{4\beta}\langle \phi+2\pi \mu,\phi+2\pi \mu\rangle}
\end{align}
which is exactly \ref{Twist_3}.

We also comment here on a harder way to derive this formula. If we wish to directly use the $\phi=0$ result, we should write \ref{Total} in the form of 
\begin{align}
{Z_\sigma }(\beta ) = c{\left( {\frac{K}{{4\pi \beta }}} \right)^{n/2}}{e^{\frac{{n\beta }}{{24K}}}}\sum\limits_{\mu  \in \Lambda (G)} {\prod\limits_{\alpha  \in {R_ + }} {\frac{{\langle \alpha ,2\pi \mu \rangle }}{{2\sin (\langle \alpha ,2\pi \mu \rangle /2)}}} } {e^{2\pi i\langle \sigma  + \rho ,\mu \rangle }}{e^{ - \frac{K}{{4\beta }}\langle 2\pi \mu ,2\pi \mu \rangle }}
\end{align}
where in this form, pole cancellation happens and thus it will turn a summation formula in a single term. Then we could apply the argument similar in $\text{U}(1)$ to shift $2\pi \mu$ by $\phi+2\pi \mu$, which will give \ref{Twist_3}.

Similarly, we also observe that with fixed chemical potential, the partition function has no contribution when $\beta \gg K$. Thus, in the following examples, we will give analysis in detail in the limit where $\beta \ll K$.
\section{Examples and properties}
\subsection{Example: $\text{SU}(M+1)$}
\textbf{Single charge sector: }In this section we are going to evaluate 
\begin{align}
{Z_\sigma }(\beta ,\mu ) = \sum\limits_{\vec n \in \mathbb{Z}_{ \ge 0}^p} F (\mu  + \rho  + \vec n \cdot \vec \alpha )~~~~~~F(\nu ) = {e^{ - \beta (|\nu {|^2} - |\rho {|^2})/K}}\prod\limits_{\alpha  \in {R_ + }} {\frac{{\langle \alpha ,\nu \rangle }}{{\langle \alpha ,\rho \rangle }}} 
\end{align}
for $\text{SU}(M+1)$. Here all the inner products are defined over the $\omega_i$ basis, where 
\begin{align}
{\omega _i} = {e_i} - {e_{i + 1}}
\end{align}
for $i=1,2,\cdots,M$. Here we have 
\begin{align}
{R_ + } = \left\{ {{\alpha _i} = {e_i} - {e_j},i < j} \right\}
\end{align}
So $\abs{R_+}=\frac{1}{2}M(M+1)=p$. In this basis, we have
\begin{align}
\rho=\rho_i \omega_i~~~~~~\rho_i=i\times (M-i+1)/2
\end{align}
And we define the inner product
\begin{align}
\left<a,b\right>=\mathcal{C}_{ij}a_ib_j~~~~~~a=a_i\omega_i~~~~~~b=b_i\omega_i
\end{align}
and the Cartan matrix
\begin{align}
\mathcal{C}_{ij} = \left\{ \begin{array}{l}
2~~~~~~i = j\\
 - 1~~~\left| {i - j} \right| = 1\\
0~~~~~~\text{others}
\end{array} \right.
\end{align}
Thus 
\begin{align}
{\left| \rho  \right|^2} = \left\langle {\rho ,\rho } \right\rangle  = \frac{1}{12}M(M + 1)(M + 2)
\end{align}
The charge sector $\mu$ is taken from $P(G)$. For $\text{SU}(M+1)$, $P(G)$ is expanded by $\kappa$ basis, we have
\begin{align}
P(G) = \left\{ {{\kappa _j}:\left\langle {{\omega _i},{\kappa _j}} \right\rangle  = {\delta _{ij}}} \right\}
\end{align}
and $\Lambda(G)$ is expanded by $\omega$ basis
\begin{align}
\Lambda (G) = \left\{ {{\omega _j}} \right\}
\end{align}
we could write down $\kappa$ in teams of $\omega$
\begin{align}
&{\kappa _i} = {X_{ij}}{\omega _j}\nonumber\\
&{X_{ij}} = \left\{ {\begin{array}{*{20}{c}}
{\frac{{i(M + 1 - j)}}{{M + 1}}}&{i < j}\\
{\frac{{j(M + 1 - i)}}{{M + 1}}}&{i \ge j}
\end{array}} \right.
\end{align}
Let us firstly consider $M=1$. In this case the charge sector is taken as $\mu\in \mathbb{Z}/2$. In the proof we make the dominate assumption, thus we take the non-negative $\mu$. So the partition function is 
\begin{align}
{Z_\sigma }(\beta ,\mu ) = \sum\limits_{n = 0}^{ + \infty } {(1 + \mu  + n)\exp \left( { - \frac{{2\beta }}{K}\left( {{{(n + \mu  + 1)}^2} - 1} \right)} \right)} 
\end{align}
Considering that $K\gg\beta $, one can estimate the result by the following integral
\begin{align}
{Z_\sigma }(\beta ,\mu ) \sim \int_{1 + \mu }^{ + \infty } {\exp \left( { - \frac{{2\beta }}{K}\left( {{x^2} - 1} \right)} \right)xdx}  = \frac{K}{{4\beta }}{e^{ - \frac{{2\beta \mu (\mu  + 2)}}{K}}}
\end{align}
We could make the following estimations here
\begin{itemize}
\item If $\mu  \ll {\left( {\frac{K}{\beta }} \right)^{1/2}}$ then the partition function is simply scales as 
\begin{align}
{Z_\sigma }(\beta ,\mu ) \sim  \frac{K}{{\beta }}
\end{align}
Here it means that $\mu$ is sufficiently closed to $\rho$. In this case, for small $n$ each exponential term in the sum is sufficiently closed to 1, then the sum is effectively
\begin{align}
{Z_\sigma }(\beta ,\mu ) \sim \sum\limits_{n = 0}^{{n^2} \approx \frac{K}{\beta }} n  \sim \frac{K}{\beta }
\end{align}
while for $n^2>\frac{K}{\beta}$ the terms are close to zero so we truncate the sum. This explain the result of the direct integral.
\item If $\mu \sim {\left( {\frac{K}{\beta }} \right)^{1/2}}$ and even larger, the partition function will exponentially decay towards zero as $\beta$ increases. The exponential decay could be explained by the following. Since $\mu$ is sufficiently large, the terms in the sum decay very fast, so the result is dominated by the first term in the sum
\begin{align}
{Z_\sigma }(\beta ,\mu ) \sim \exp \left( { - \frac{{2\beta }}{K}{\mu ^2}} \right)
\end{align}
and the power law decaying factor is no longer important. 
\end{itemize}
Motivated by the discussions in $\text{SU}(2)$, we could make a generic estimation on the $\text{SU}(M+1)$ result. Firstly take a look on the structure of $F(\nu)$, the set $R_+$ contains $\frac{M(M+1)}{2}$ terms, and each term in 
\begin{align}
\frac{{\left\langle {\alpha ,\nu } \right\rangle }}{{\left\langle {\alpha ,\rho } \right\rangle }}
\end{align}
is a linear sum of $\nu_i$. So we could write the partition function as
\begin{align}
{Z_\sigma }(\beta ,\mu )\sim\exp \left( {\frac{\beta }{K}{{\left| \rho  \right|}^2}} \right)\sum\limits_{\{ q\} } {{c_{\{ q\} }}} \prod\nolimits_{i = 1}^{M(M +1)/2} {\left( {\int_{{\nu _i} > \rho  + \mu } {d{\nu _i}\nu _i^{{q_i}}\exp \left( { - \frac{\beta }{K}\nu _i^2} \right)} } \right)} 
\end{align}
where $c_{\{q\}}$ are coefficients for set $\{q_i\}$ satisfying
\begin{align}
q_i\in \mathbb{Z}_{\ge 0}~~~~~~\sum\limits_{i = 1}^{M(M+1)/2} {{q_i}}  \le \frac{{M(M + 1)}}{2}
\end{align}
Thus here it is convenient to define the following function
\begin{align}
{g_q}(t,z) = \int_z^{ + \infty } {{x^q}\exp \left( { - \frac{{{x^2}}}{t}} \right)dx}  = \frac{1}{2}{z^{1 + q}}{E_{(1 - q)/2}}(\frac{{{z^2}}}{t})
\end{align}
where $E_n(z)$ is the standard exponential integral function
\begin{align}
{E_n}(z) = \int_1^{ + \infty } {\frac{{{e^{ - zt}}}}{{{t^n}}}dt} 
\end{align}
So we have
\begin{align}
{Z_\sigma }(\beta ,\mu )\sim \exp \left( {\frac{\beta }{K}\frac{{M(M + 1)(M + 2)}}{12}} \right)\sum\limits_{\{ q\} } {{c_{\{ q\} }}} \prod\nolimits_i {{g_{{q_i}}}(\frac{K}{\beta },{\rho _i} + {\mu _i})} 
\end{align}
Here we could make the following treatment, and there are three following numbers that are possibly large: $\frac{K}{\beta}$, $M$ and $\mu$.
\begin{itemize}
\item Large $\frac{K}{\beta}$, relatively small $M$ and $\mu$. Here we use the expansion
\begin{align}
{g_q}(z,t) \sim \frac{1}{2}{t^{\frac{{q + 1}}{2}}}\Gamma \left( {\frac{{q + 1}}{2}} \right)
\end{align}
Thus larger $q_i$ means a dominated decaying rate. So we take
\begin{align}
\sum\limits_{i = 1}^{M(M+1)/2} {{q_i}}= \frac{{M(M + 1)}}{2}
\end{align}
where for thus terms we have
\begin{align}
\sum\limits_{i = 1}^{M(M+1)/2}  {\frac{{{q_i} + 1}}{2}}  = \frac{{M(M + 1)}}{2}
\end{align}
So we get
\begin{align}
{Z_\sigma }(\beta ,\mu ) \sim {\left( {\frac{K}{\beta }} \right)^{\frac{{(M + 1)M}}{2}}}
\end{align}
\item Large $\frac{K}{\beta}$, $M$ is kept to be small, but some $\mu_i$ is sufficiently large, $\mu_i^2\sim \frac{K}{\beta}$, then that will cause a fast exponential decay. We have
\begin{align}
{Z_\sigma }(\beta ,\mu ) \sim {e^{ - \#  \times \frac{{\beta \mu _i^2}}{K}}}
\end{align}
\item For $\text{poly}(M)\sim \frac{K}{\beta} \gg 1$, the decaying rate will increase dramatically. For small $\mu$ case, the Gamma function will provide exponential decay about polynomials of $M$. Moreover, the exponential decaying part will be $M$-fold. For large $\mu$ case, we also get a $M$-fold exponential decaying.
\end{itemize}
\textbf{Whole sector: }Now we take a look at the whole sector. The result is given by
\begin{align}
{Z_\sigma }(\beta |\phi ) = c{\left( {\frac{K}{{4\pi \beta }}} \right)^{M(M + 2)/2}}{e^{\frac{{M(M + 2)\beta }}{{24K}}}}{\Theta _\sigma }(\beta ,\phi )
\end{align}
where
\begin{align}
c = {(2\pi )^{p + r}}{(\text{det} {\cal C})^{1/2}}\prod\limits_{\alpha  \in {R_ + }}^{} {\frac{1}{{\left\langle {\alpha ,\rho } \right\rangle }}} 
\end{align}
and
\begin{align}
{\Theta _\sigma }(\beta ,\phi ) = \sum\limits_{\mu  \in \Lambda (G)} {{e^{2\pi i\langle \sigma  + \rho ,\mu \rangle }}{e^{ - \frac{K}{{4\beta }}\langle \phi  + 2\pi \mu ,\phi  + 2\pi \mu \rangle }}\prod\limits_{\alpha  \in {R_ + }} {\frac{{\langle \alpha ,\phi  + 2\pi \mu \rangle }}{{2\sin (\langle \alpha ,\phi  + 2\pi \mu \rangle /2)}}} } 
\end{align}
Thus in $\text{SU}(M+1)$, we have 
\begin{align}
&\text{det}\mathcal{C}=M+1\nonumber\\
&p=\frac{1}{2}M(M+1)~~~~~~r=M\nonumber\\
&c = {(2\pi )^{M(M + 3)/2}}{(M + 1)^{1/2}}\frac{1}{{{2^{{M^2}}}}}
\end{align}
In $\text{SU}(M+1)$, since it is simply connected, thus $\sigma$ is taking on arbitrary element from $P(G)$. Thus, the result is not actually related to the spin structure
\begin{align}
{\Theta _\sigma }(\beta ,\phi ) = \sum\limits_{\mu  \in \Lambda (G)} {{e^{2\pi i\langle \rho ,\mu \rangle }}{e^{ - \frac{K}{{4\beta }}\langle \phi  + 2\pi \mu ,\phi  + 2\pi \mu \rangle }}\prod\limits_{\alpha  \in {R_ + }} {\frac{{\langle \alpha ,\phi  + 2\pi \mu \rangle }}{{2\sin (\langle \alpha ,\phi  + 2\pi \mu \rangle /2)}}} } 
\end{align}
Moreover, since $\langle \rho ,\mu \rangle $ is always integer, the phase term should also be removed and the result is
\begin{align}
{\Theta _\sigma }(\beta ,\phi ) = \sum\limits_{\mu  \in \Lambda (G)} {{e^{ - \frac{K}{{4\beta }}\langle \phi  + 2\pi \mu ,\phi  + 2\pi \mu \rangle }}\prod\limits_{\alpha  \in {R_ + }} {\frac{{\langle \alpha ,\phi  + 2\pi \mu \rangle }}{{2\sin (\langle \alpha ,\phi  + 2\pi \mu \rangle /2)}}} } 
\end{align}
For $\text{SU}(2)$ we have the fact that $\phi$ and $\sigma$ are numbers.
\begin{align}
\prod\limits_{\alpha  \in {R_ + }} {\frac{{\langle \alpha ,\phi  + 2\pi \mu \rangle }}{{2\sin (\langle \alpha ,\phi  + 2\pi \mu \rangle /2)}}}  = \frac{{\phi  + 2\pi \mu }}{{\sin \phi }}
\end{align}
where $\mu$ is integer. So we get
\begin{align}
{\Theta _\sigma }(\beta ,\phi ) = \sum\limits_{\mu  \in \mathbb{Z}} {{e^{ - \frac{K}{{2\beta }}{{(\phi  + 2\pi \mu )}^2}}}(\phi  + 2\pi \mu )\frac{1}{{\sin (\phi )}}} 
\end{align}
Now we notice that, if $\phi=2\pi\mathbb{Z}$, let's say $\phi=2\pi n_\phi$ where $n_\phi$ is an integer. Then we say that most terms in the sum cancel except $\mu=-n_\phi$, which gives
\begin{align}
{\Theta _\sigma }(\beta ,\phi ) = 1
\end{align}
So the result is simply
\begin{align}
{Z_\sigma }(\beta |\phi ) \sim {\left( {\frac{K}{\beta }} \right)^{3/2}}{e^{\frac{\beta }{{12K}}}} \sim {\left( {\frac{K}{\beta }} \right)^{3/2}}
\end{align}
Another case is that $\phi$ is not in $2\pi\mathbb{Z}$, then write $\text{Round}(x)$ the integer closest to $x$, then the sum is dominated by
\begin{align}
&{\Theta _\sigma }(\beta ,\phi ) \sim {e^{ - \frac{K}{{2\beta }}{{(\phi  - 2\pi {\rm{Round}}\left( {\frac{\phi }{{2\pi }}} \right))}^2}}}(\phi  - 2\pi {\rm{Round}}\left( {\frac{\phi }{{2\pi }}} \right))\frac{1}{{\sin (\phi )}}\nonumber\\
&\sim {e^{ - \frac{K}{{2\beta }}{{(\phi  - 2\pi {\rm{Round}}\left( {\frac{\phi }{{2\pi }}} \right))}^2}}}(\phi  - 2\pi {\rm{Round}}\left( {\frac{\phi }{{2\pi }}} \right))\frac{1}{{\sin (\phi )}}
\end{align}
Thus, in general, the result will get an exponential decay
\begin{align}
{Z_\sigma }(\beta |\phi )\sim{\left( {\frac{K}{\beta }} \right)^{3/2}}{e^{ - \frac{K}{{2\beta }}{{(\phi  - 2\pi {\rm{Round}}\left( {\frac{\phi }{{2\pi }}} \right))}^2}}}
\end{align}
Now we consider generic $\text{SU}(M+1)$ case. The result is similar. We have
\begin{itemize}
\item For given $\phi$, If there exists $\mu \in \Lambda(G)$ and $\alpha \in R_+$ such that 
\begin{align}
{\langle \alpha ,\phi  + 2\pi \mu \rangle }=0
\end{align}
then there might be multiple solutions of $\alpha$ and $\mu$ for that given $\phi$. Find all of them, and we get a set of allowed $\mu$. Then the partition function scales as
\begin{align}
{Z_\sigma }(\beta |\phi )\sim{\left( {\frac{K}{\beta }} \right)^{M(M + 2)/2}}{e^{\frac{{M(M + 2)\beta }}{{24K}}}}{e^{ - \frac{K}{{4\beta }}{\rm{Mi}}{{\rm{n}}_{{\rm{allowed }}\mu }}\left( {\langle \phi  + 2\pi \mu ,\phi  + 2\pi \mu \rangle } \right)}}
\end{align}
It is possible that we could have
\begin{align}
{\rm{Min}}{{\rm{}}_{{\text{allowed }}\mu }}\left( {\langle \phi  + 2\pi \mu ,\phi  + 2\pi \mu \rangle } \right) = 0
\end{align}
where in this case we get
\begin{align}
{Z_\sigma }(\beta |\phi )\sim{\left( {\frac{K}{\beta }} \right)^{M(M + 2)/2}}{e^{\frac{{M(M + 2)\beta }}{{24K}}}}
\end{align}
\item If it does not exist such $\mu$, we have
\begin{align}
{Z_\sigma }(\beta |\phi ) \sim {\left( {\frac{K}{\beta }} \right)^{M(M + 2)/2}}{e^{\frac{{M(M + 2)\beta }}{{24K}}}}{e^{ - \frac{K}{{4\beta }}{\text{Min}_\mu }\left( {{{\langle \phi  + 2\pi \mu ,\phi  + 2\pi \mu \rangle }}} \right)}}
\end{align}
\end{itemize}
Thus, there is an interesting bound we could find for every possible $\phi$. Since we know that 
\begin{align}
 - \mathcal{O}(1) \times M \lessapprox  - {\rm{Mi}}{{\rm{n}}_\mu }\left( {\langle \phi  + 2\pi \mu ,\phi  + 2\pi \mu \rangle } \right) \lessapprox 0
\end{align}
Where $\mathcal{O}(1)$ means a numerical constant. Thus we know that
\begin{align}
{\left( {\frac{K}{\beta }} \right)^{M(M + 2)/2}}{e^{\frac{{M(M + 2)\beta }}{{24K}}}}{e^{ - \frac{{KM}}{{4\beta }} \times \mathcal{O}(1)}} \lessapprox {Z_\sigma }(\beta |\phi ) \lessapprox {\left( {\frac{K}{\beta }} \right)^{M(M + 2)/2}}{e^{\frac{{M(M + 2)\beta }}{{24K}}}}
\end{align}
The bound in the RHS could appear in any $M$. In fact, in $\text{SU}(M+1)$, we consider $\phi$ to be zero, one can show that there is one single $\mu=0$ to make the function $ {\langle \phi  + 2\pi \mu ,\phi  + 2\pi \mu \rangle } $ to get minimized at zero according to the assumption that existing $\alpha$ and $\mu$ such that ${\langle \alpha ,\phi  + 2\pi \mu \rangle }=0$. This term gives the contribution to the $\Theta$ function $\mathcal{O}(1)$, thus we get
\begin{align}
{Z_\sigma }(\beta |\phi )\sim{\left( {\frac{K}{\beta }} \right)^{M(M + 2)/2}}{e^{\frac{{M(M + 2)\beta }}{{24K}}}}
\end{align}

\subsection{$\text{SO}(2M+1)$}
\textbf{{\text{SO}(3)}: } The most simplest case, $\text{SO}(3)$, is very similar with $\text{SU}(2)$, where we have computed before, and thus it is slightly different from the general $\text{SO}(2M+1)$ case with $M\ge 2$. So we discuss it separately. 

Firstly, the similarities are that we have the same $\omega$ basis
\begin{align}
{\omega _1} = {e_1} - {e_2}
\end{align}
and the same matrix $\mathcal{C}$
\begin{align}
\mathcal{C}_{11}=2
\end{align}
and we have $\Lambda(G)=\omega_1\mathbb{Z}$, $\phi\in \mathbb{R}$. We also have the same $R_+$, $R_+=\{\omega_1\}$, $\abs{R_+}=1=p$, and $\rho=\omega_1/2$, $r=1$, so the constant $c=2\sqrt{2}\pi^2$. And $P(G)$ should be the dual lattice of $\Lambda(G)$, namely
\begin{align}
P(G) = \frac{1}{2}\mathbb{Z}{\omega _1}
\end{align}
and $\mu \in P(G)+\sigma_i$ where we assume $\mu$ is dominate (non-negative), and depending on the spin structure $\sigma_i$. There are two spin structures:
\begin{align}
\sigma_0=0~~~~~~\sigma_1=\frac{1}{4}\omega_1
\end{align}
The single charge sector result is exactly the same as $\text{SU}(2)$
\begin{align}
{Z_\sigma }(\beta ,\mu ) \sim \frac{K}{{4\beta }}{e^{ - \frac{{2\beta \mu (\mu  + 2)}}{K}}}
\end{align}
although the choice of $\mu$ is different. Moreover, since
\begin{align}
&{\Theta _{{\sigma _0}}}(\beta ,\phi ) = \sum\limits_{\mu  \in \mathbb{Z}} {{e^{ - \frac{K}{{2\beta }}{{(\phi  + 2\pi \mu )}^2}}}(\phi  + 2\pi \mu )\frac{1}{{\sin (\phi )}}} \nonumber\\
&{\Theta _{{\sigma _1}}}(\beta ,\phi ) = \sum\limits_{\mu  \in \mathbb{Z}} {{{( - 1)}^\mu }{e^{ - \frac{K}{{2\beta }}{{(\phi  + 2\pi \mu )}^2}}}(\phi  + 2\pi \mu )\frac{1}{{\sin (\phi )}}} 
\end{align}
Thus we see that the non-trivial spin structure only brings the phase factor, and thus in the limit we are interested in, those two spin structures are both approximately give the same expressions as the whole sector formula of $\text{SU}(2)$. 
\\
\\
\textbf{Single charge sector: }We will list the necessary data for $\text{SO}(2M+1)$ here.
\begin{align}
&{\omega _i} = \left\{ {\begin{array}{*{20}{l}}
{{e_i} - {e_{i + 1}}}&{i = 1, \cdots ,M - 1}\\
{{e_M}}&{i = M}
\end{array}} \right.\nonumber\\
&{R_ + } = \left\{ {{\alpha _i} = {e_i} \pm {e_j},i < j:{\text{for  }}i,j = 1,2, \cdots ,M} \right\} \cup \left\{ {{e_i}:{\text{for  }}i = 1,2, \cdots ,M} \right\}\nonumber\\
&\left| {{R_ + }} \right| = {M^2} = p\nonumber\\
&\rho  =(2(M-i)+1)e_i\nonumber\\
&{\mathcal{C}_{ij}} = \left\{ {\begin{array}{*{20}{l}}
{2\;\;\;\;\;\;i = j < M}\\
{ - 1\;\;\;\left| {i - j} \right| = 1,i \ne M}\\
{ - 2\;\;\;i = M,j = M - 1}\\
{0\;\;\;\;\;\;{\rm{others}}}
\end{array}} \right.\nonumber\\
&{\left| \rho  \right|^2} = \left\langle {\rho ,\rho } \right\rangle  =\frac{1}{{12}}M(M + 1)(4M - 1)
\end{align}
The lattice is defined as
\begin{align}
&\Lambda (G) = \{ \sum\limits_{j = 1}^M {{\lambda _j}{\omega _j}} :{\lambda _j} \in \mathbb{Z}\} \nonumber\\
&P(G) = \{ \sum\limits_{j = 1}^M {{\lambda _j}{e_j}} :{\lambda _i} \in \mathbb{Z}/2,{\lambda _i} - {\lambda _j} \in \mathbb{Z}\} 
\end{align}
and $\phi\in \mathbb{R}^M$ and $\mu \in P(G)+\sigma_i$ (where $\mu$ is chosen to be dominate). In any $M$, we have two possible spin structures, $\sigma_0=0$ and $\sigma_1=\omega_{M-1}/2$. 

The result for single charge sector is pretty similar with $\text{SU}(M+1)$ case. We have
\begin{itemize}
\item Large $\frac{K}{\beta}$, relatively small $M$ and $\mu$. Since we know that 
\begin{align}
\sum\limits_{i = 1}^{M^2} {{q_i}}  = {M^2}
\end{align}
thus
\begin{align}
\sum\limits_{i = 1}^{M^2} {\frac{{{q_i} + 1}}{2}}  = M^2
\end{align}
So we get
\begin{align}
{Z_\sigma }(\beta ,\mu ) \sim {\left( {\frac{K}{\beta }} \right)^{M^2}}
\end{align}
\item Large $\frac{K}{\beta}$, $M$ is kept to be small, but some $\mu_i$ is sufficiently large, $\mu_i^2\sim \frac{K}{\beta}$, we have
\begin{align}
{Z_\sigma }(\beta ,\mu ) \sim {e^{ - \#  \times \frac{{\beta \mu _i^2}}{K}}}
\end{align}
\item Large $M$ will provide even faster decaying rate.
\end{itemize}
\textbf{Whole sector: }The corresponding data is
\begin{align}
&\text{det}\mathcal{C}=2~~~~~~n=(2M+1)M\nonumber\\
&p=M^2~~~~~~r=M\nonumber\\
&c ={(2\pi )^{{M^2} + M}}\sqrt 2\left( \prod\nolimits_{k = 0}^{M - 1} {{{\left( {(2k + 2)(2k + 3)} \right)}^{M - 1 - k}}} \right)^{-1}
\end{align}
We could write down the $\Theta$ function as 
\begin{align}
{\Theta _{\sigma _i}}(\beta ,\phi ) = \sum\limits_{\mu  \in \Lambda (G)} {{e^{2\pi i\langle \rho+\sigma_i ,\mu \rangle }}{e^{ - \frac{K}{{4\beta }}\langle \phi  + 2\pi \mu ,\phi  + 2\pi \mu \rangle }}\prod\limits_{\alpha  \in {R_ + }} {\frac{{\langle \alpha ,\phi  + 2\pi \mu \rangle }}{{2\sin (\langle \alpha ,\phi  + 2\pi \mu \rangle /2)}}} } 
\end{align}
There are some differences between the $\text{SO}(2M+1)$ case and the $\text{SU}(M+1)$ case we discuss above. Since it has two different spin structures, they may change signatures in different terms of partition function sum. Let us take a look on $M=2$ case for example. In this case, we assume $\phi=2\pi\eta _i\omega_i$ and $\mu=\mu_i\omega_i$. We have
\begin{align}
&2\pi i\left\langle  {\rho  + {\sigma _0},{\mu}} \right\rangle  = -2\pi i{\mu _1} + 5\pi i{\mu _2}\sim\pi i \mu_2\nonumber\\
&2\pi i\left\langle { \rho  + {\sigma _1},\mu} \right\rangle  = 4\pi i{\mu _1}\sim0
\end{align}
where here $\sim$ means modulo $2\pi i\mathbb{Z}$. Thus we know that for ${\sigma _1}$ the phase is always 1, while for $\sigma_0$ it depends on $\mu_2$ is even or odd.

Similarly, considering $M=2$ and for simplicity we assume $\sigma=\sigma_0$. If there is no solution for the following equation
\begin{align}
{{\langle \alpha ,\phi  + 2\pi \mu \rangle }}=0
\end{align}
namely, there is no solution for any of the following equations
\begin{align}
&{\eta _2} + {\mu _2} = 0\nonumber\\
&- 2{\eta _1} + {\eta _2} - 2{\mu _1} + {\mu _2} = 0\nonumber\\
&- {\eta _1} + {\eta _2} - {\mu _1} + {\mu _2} = 0\nonumber\\
&- 2{\eta _1} + 3{\eta _2} - 2{\mu _1} + 3{\mu _2} = 0
\end{align}
then 
\begin{align}
{Z_\sigma }(\beta |\phi ) \sim   {\left( {\frac{K}{\beta }} \right)^5}{e^{\frac{{5\beta }}{{12K}}}}{e^{ - \frac{K}{{4\beta }}\text{Min}{_\mu }\left\langle {\phi  + 2\pi \mu ,\phi  + 2\pi \mu } \right\rangle }}
\end{align}
If there exists solution, for instance, say if $\eta_2\in \mathbb{Z}$, then consider the solution $\eta_2=-\mu_2$ we get
\begin{align}
&\prod\limits_{\alpha  \in {R_ + }} {\frac{{\langle \alpha ,\phi  + 2\pi \mu \rangle }}{{2\sin (\langle \alpha ,\phi  + 2\pi \mu \rangle /2)}}}  \sim \frac{{{{(\pi ({\eta _1} + {\mu _1}))}^3}}}{{{{\sin }^3}(2\pi ({\eta _1} + {\mu _1}))}}\nonumber\\
&\frac{K}{{4\beta }}\left\langle {\phi  + 2\pi \mu ,\phi  + 2\pi \mu } \right\rangle  = \frac{2K}{\beta }{(\pi ({\eta _1} + {\mu _1}))^2}
\end{align}
Then if $\eta_2$ is not integer we have
\begin{align}
{Z_\sigma }(\beta |\phi )\sim \pm {\left( {\frac{K}{\beta }} \right)^5}{e^{\frac{{5\beta }}{{12K}}}}{e^{ - \frac{2K}{\beta }\left( {\pi ({\eta _2} - {\rm{Round}}({\eta _2}))} \right)}}^2
\end{align}
If $\eta_2$ is an integer, we have
\begin{align}
{Z_\sigma }(\beta |\phi )\sim \pm {\left( {\frac{K}{\beta }} \right)^5}{e^{\frac{{5\beta }}{{12K}}}}
\end{align}
These facts will happen in general, where the generic form is expected to be
\begin{align}
{Z_\sigma }(\beta |\phi )\sim {\left( {\frac{K}{\beta }} \right)^{(2M + 1)M/2}}{e^{\frac{{(2M + 1)M\beta }}{{24K}}}}{e^{ - \frac{K}{{4\beta }}{\rm{Mi}}{{\rm{n}}_\mu }\left( {\langle \phi  + 2\pi \mu ,\phi  + 2\pi \mu \rangle } \right)}}
\end{align}
where if we get poles as described above, the minimizing function will only localized on those poles. And we also have a similar bound
\begin{align}
{\left( {\frac{K}{\beta }} \right)^{(2M + 1)M/2}}{e^{\frac{{(2M + 1)M\beta }}{{24K}}}}{e^{ - \frac{{KM}}{{4\beta }} \times \mathcal{O}(1)}} \lessapprox {Z_\sigma }(\beta |\phi ) \lessapprox {\left( {\frac{K}{\beta }} \right)^{(2M + 1)M/2}}{e^{\frac{{(2M + 1)M\beta }}{{24K}}}}
\end{align}
\subsection{$\text{SO}(2M)$}
\textbf{Single charge sector: }We list the data we need to use here
\begin{align}
&{\omega _i} = \left\{ {\begin{array}{*{20}{c}}
{{e_i} - {e_{i + 1}}}&{1 \le i \le M - 1}\\
{{e_{M - 1}} + {e_M}}&{i = M}
\end{array}} \right.\nonumber\\
&{R_ + } = \left\{ {{e_i} \pm {e_j}:i < j} \right\}~~~~~~\left| {{R_ + }} \right| = p = M(M - 1)\nonumber\\
&\rho  = 2(M - j){e_j}~~~~~~{\left| \rho  \right|^2} = \frac{1}{6}(M - 1)M(2M - 1)\nonumber\\
&{{\mathcal{C}}_{ij}} = \left\{ {\begin{array}{*{20}{l}}
{2\;\;\;\;\;\;i = j < M}\\
{ - 1\;\;\;\left| {i - j} \right| = 1,i,j \ne M}\\
{ - 1\;\;\;(i,j) = (M,M - 2)\text{ or }(M - 2,M)}\\
{0\;\;\;\;\;\;{\rm{others}}}
\end{array}} \right.
\end{align}
and we know the lattices are 
\begin{align}
&\Lambda (G) = \{ \sum\limits_{j = 1}^M {{\lambda _j}{\omega _j}} :{\lambda _j} \in \mathbb{Z}\} \nonumber\\
&P(G) = \{ \sum\limits_{j = 1}^M {{\lambda _j}{e_j}} :{\lambda _i} \in \mathbb{Z}/2,{\lambda _i} - {\lambda _j} \in \mathbb{Z}\} 
\end{align}
and again, we know that $\phi\in \mathbb{R}^M$ and $\mu\in P(G)_{\ge 0}$. Moreover, we have four spin structures 
\begin{align}
&\sigma_0=0~~~~~~{\sigma _1} = \sum\limits_{i = 1}^M {\frac{1}{2}{\omega _i}} \nonumber\\
&{\sigma _2} = {\omega _1}~~~~~{\sigma _3} = \left\{ {\begin{array}{*{20}{c}}
{ - \frac{1}{2}{\omega _1} + \sum\limits_{i = 2}^M {\frac{1}{2}{\omega _i}} }&{M\text{ is even}}\\
{\sum\limits_{i = 2}^M {\frac{3}{2}{\omega _i}} }&{M\text{ is odd}}
\end{array}} \right.
\end{align}
where when $M$ is even the spin structures form the group $\mathbb{Z}_2\times \mathbb{Z}_2$, while when $M$ is odd the spin structures form the group $\mathbb{Z}_4$.

The above discussion has the restriction that $M\ge 2$. For the reduced case $M=1$, $\text{SO}(2)=\text{U}(1)$, thus it is not semisimple, we will discuss it later.

The generic feature of the result for single charge sector is the same as before. The only difference is that now for large $\frac{K}{\beta}$ but relatively small $\mu$ and $M$, since
\begin{align}
\sum\limits_{i = 1}^{ M(M - 1)} {{q_i}}  = M(M - 1)
\end{align}
so 
\begin{align}
\sum\limits_{i = 1}^{ M(M - 1)} {\frac{{{q_i} + 1}}{2}}  = M(M - 1)
\end{align}
Thus the partition function in this limit is given by
\begin{align}
{Z_\sigma }(\beta ,\mu ) \sim {\left( {\frac{K}{\beta }} \right)^{ M(M - 1)}}
\end{align}
\\
\\
\textbf{Whole sector: }The constants are
\begin{align}
&\det \mathcal{C} = 4~~~~~~n = (2M - 1)M\nonumber\\
&p = M(M - 1)~~~~~~r = M\nonumber\\
&c = 2{(2\pi )^{{M^2}}}{\left( {(M - 1)!\prod\nolimits_{s = 1,3,5 \ldots }^{2M - 3} {s!} } \right)^{ - 1}}
\end{align}
One of the main differences comparing to previous cases is that now we are four spin structures. Generically, we conclude that more spin structures may lead to more complicated cases in the phases of terms for summation.

The form of the partition function is
\begin{align}
{Z_\sigma }(\beta |\phi )\sim{\left( {\frac{K}{\beta }} \right)^{(2M - 1)M/2}}{e^{\frac{{(2M - 1)M\beta }}{{24K}}}}{e^{ - \frac{K}{{4\beta }}{\rm{Mi}}{{\rm{n}}_\mu }\left( {\langle \phi  + 2\pi \mu ,\phi  + 2\pi \mu \rangle } \right)}}
\end{align}
with a similar bound
\begin{align}
{\left( {\frac{K}{\beta }} \right)^{(2M - 1)M/2}}{e^{\frac{{(2M - 1)M\beta }}{{24K}}}}{e^{ - \frac{{KM}}{{4\beta }} \times {\cal O}(1)}} \mathbin{\lower.3ex\hbox{$\buildrel<\over
{\smash{\scriptstyle\approx}\vphantom{_2}}$}} {Z_\sigma }(\beta |\phi ) \mathbin{\lower.3ex\hbox{$\buildrel<\over
{\smash{\scriptstyle\approx}\vphantom{_2}}$}} {\left( {\frac{K}{\beta }} \right)^{(2M - 1)M/2}}{e^{\frac{{(2M - 1)M\beta }}{{24K}}}}
\end{align}
\subsection{Non-semisimple cases}
For non-semisimple case, a standard example is $\text{U}(M+1)$ ($M \ge 0$). Formally, the theorems above are not applied for generic non-semisimple groups, but since  $\text{U}(M+1)$ is a combination of $\text{U}(1)$ and $\text{SU}(M+1)$, we could still use it practically by merging the result of $\text{U}(1)$ and $\text{SU}(M+1)$ together. 

We will revisit the single charge sector and the whole sector partition functions for $\text{U}(1)$ in the following, and make some predictions for $\text{U}(M+1)$ in general.
\\
\\
\textbf{$\text{U}(1)$: }We revisit our $\text{U}(1)$ case in our mathematical framework. As we know, $\text{U}(1)$ has exactly two spin structures: the trivial one $\sigma_0$, and the M\"obius $\sigma_1$. If one identifies the weight lattice $P(\text{U}(1))$ of $\text{U}(1)$ with $\mathbb Z$, and the inner product $\langle -,-\rangle$ is just multiplication of numbers, then $\sigma_0$ can be chosen to be represented by 1, and $\sigma_1$ can be chosen to be represented by $1/2$. Note that since there is no semisimple component in $\text{U}(1)$, $R_+$ is an empty set.

Now applying the single charge sector formula we directly obtain
\begin{align}
{Z_\sigma }(\beta ,\mu ) = \exp ( - \frac{{\beta {\mu ^2}}}{K})
\end{align}
which precisely matches our previous observation. Secondly, for the whole partition function we have
\begin{align}
&{Z_{{\sigma _0}}}(\beta |\phi ) \sim {(\frac{K}{\beta })^{1/2}}\sum\limits_\mu  {{e^{ - \frac{K}{{4\beta }}{{(\phi  + 2\pi \mu )}^2}}}} \nonumber\\
&{Z_{{\sigma _1}}}(\beta |\phi ) \sim {(\frac{K}{\beta })^{1/2}}\sum\limits_\mu  {{{( - 1)}^\mu }{e^{ - \frac{K}{{4\beta }}{{(\phi  + 2\pi \mu )}^2}}}} 
\end{align}
which matches our results before.
\\
\\
\textbf{$\text{U}(M+1)$: }The spin structures of $\text{U}(M+1)$ are non-trivial. Generically, there are two spin structures for generic $M$. In fact, since
\begin{align}
{\pi _1}({\rm{U}}(M+1)) = {\pi _1}({\rm{U}}(1)) = \mathbb{Z}
\end{align}
thus the spin structure for $\text{U}(M)$ is 
\begin{align}
\text{Hom}({\pi _1}({\rm{U}}(M+1)),{\mathbb{Z}_2}) = {\mathbb{Z}_2}
\end{align}

For other group data, we should merge $\text{U}(1)$ and $\text{SU}(M+1)$ together. For instance, consider $\text{U}(2)$. The structure of the group is $\text{U}(2)=\text{SU}(2)\times \text{U}(1)/\mathbb{Z}_2$. Thus, we may effectively take the products of lattices and Cartan matrices, considering the equivalence relationship provided by $\mathbb{Z}_2$. Thus, for instance, the positive simple roots are given by
\begin{align}
{R_ + } = \left\{ {({\omega _1},0)} \right\}
\end{align}
where $\omega_1$ is from $\text{SU}(2)$. The Cartan matrix is 
\begin{align}
\mathcal{C} = \text{diag}(2,1)
\end{align}
In the merging procedure we know that for generic $M$, 
\begin{align}
p = \left| {R_ + ^{{\rm{U}}(M + 1)}} \right| = \left| {R_ + ^{{\rm{SU}}(M + 1)}} \right| = \frac{{M(M + 1)}}{2}
\end{align}
Thus, the scaling of the partition function in the single charge sector, for $\text{U}(M+1)$, is the same as $\text{SU}(M+1)$, since $\sum_i q_i$ is bounded by the same number. For the whole sector, we have the bound
\begin{align}
{\left( {\frac{K}{\beta }} \right)^{{{(M + 1)}^2}/2}}{e^{\frac{{{{(M + 1)}^2}\beta }}{{24K}}}}{e^{ - \frac{{KM}}{{4\beta }} \times {\cal O}(1)}} \mathbin{\lower.3ex\hbox{$\buildrel<\over
{\smash{\scriptstyle\approx}\vphantom{_2}}$}} {Z_\sigma }(\beta |\phi ) \mathbin{\lower.3ex\hbox{$\buildrel<\over
{\smash{\scriptstyle\approx}\vphantom{_2}}$}} {\left( {\frac{K}{\beta }} \right)^{{{(M + 1)}^2}/2}}{e^{\frac{{{{(M + 1)}^2}\beta }}{{24K}}}}
\end{align}
\subsection{Generic features}\label{fe}
After going through the explicit examples of groups, we could summarize some generic features of our result. 
\\
\\
\textbf{Product manifolds: }For a product manifold $M_1\times M_2$, the partition function is also a product $Z_{M_1\times M_2}=Z_{M_1}Z_{M_2}$. This can be seen easily from the fact that the Laplace-Beltrami operator on the product manifold is the summation:
\begin{align}
\Delta_{M_1\times M_2}=\Delta_{M_1}+\Delta _{M_2}
\end{align}
so eigenfunctions are products of eigenfunctions and eigenvalues are summations of eigenvalues
\begin{align}
\psi_{n,m}=\psi_n\psi_m\quad~~~~~~\quad E_{n,m}=E_n+E_m
\end{align}
Plugging this into equation \ref{Expansion}, we get 
\begin{align}
Z_{M_1\times M_2}&=\sum_{n,m}\psi_{n,m}(\mathbf{1})\bar{\psi}_{n,m}(\mathbf{1})e^{-2\pi \beta E_{n,m}}\nonumber\\
&=\sum_{n,m}\psi_{n}(\mathbf{1})\psi_{m}(\mathbf{1})\bar{\psi}_{n}(\mathbf{1})\bar{\psi}_{m}(\mathbf{1})e^{-2\pi \beta (E_n+E_m)}\nonumber\\
&=Z_{M_1}Z_{M_2}
\end{align}
Those formulas should work for single charge sector partition functions. One could use resummation formula to obtain the whole sector result as above. 
\\
\\
\textbf{Triviality for $K\ll \beta$}: In this limit all partition functions reduce to $\mathcal{O}(1)$ constants.
\\
\\
\textbf{Single charge sector for $K\gg \beta$: }Generically we expect the following results,
\begin{itemize}
\item For group $G$ with small dimension and rank, and small absolute values of charge sectors, Given the number of positive roots $\abs{R_+}=p$, the partition function is expected to be
\begin{align}
{Z_\sigma }(\beta ,\mu ) \sim {\left( {\frac{K}{\beta }} \right)^{p}}
\end{align}
Here we should note that $\text{U}(1)$ also follows from this formula, since for $\text{U}(1)$, $p=0$, and in that limit we have ${Z_\sigma }(\beta ,\mu ) \sim 1$.
\item For small $M$, but some $\mu_i$ is comparable to $\frac{K}{\beta}$, we expect an exponential decay
\begin{align}
{Z_\sigma }(\beta ,\mu ) \sim \exp \left( { - \# \frac{{\beta \mu _i^2}}{K}} \right)
\end{align}
\item Large $M$ will make the decay rate larger.
\item The result for semisimple group $G$ is not related to spin structure $\sigma$. 
\end{itemize}
\textbf{Whole sector for $K\gg \beta$: }Generically we expect the following results,
\begin{itemize}
\item Generically, for group $G$, the result is expected to be,
\begin{align}
{Z_\sigma }(\beta |\sigma ) \sim {\left( {\frac{K}{\beta }} \right)^{n/2}}{e^{\frac{{n\beta }}{{24K}}}}{e^{ - \frac{K}{{4\beta }}{\rm{Mi}}{{\rm{n}}_\mu }\left( {\left\langle {\phi  + 2\pi \mu ,\phi  + 2\pi \mu } \right\rangle } \right)}}
\end{align}
where minimization works on the lattice $\Lambda(G)$, and the dimension of the group is given by $n$. In the case that the following equation
\begin{align}
\left\langle {\alpha ,\phi  + 2\pi \mu } \right\rangle =0
\end{align}
has solutions, where $\alpha$ is from positive roots $R_+$, the minimization is taken only over those solutions.
\item Thus we generically have a bound
\begin{align}
{\left( {\frac{K}{\beta }} \right)^{n/2}}{e^{\frac{{n\beta }}{{24K}}}}{e^{ - \frac{{Kr}}{{4\beta }} \times {\cal O}(1)}} \mathbin{\lower.3ex\hbox{$\buildrel<\over
{\smash{\scriptstyle\approx}\vphantom{_2}}$}} {Z_\sigma }(\beta |\phi ) \mathbin{\lower.3ex\hbox{$\buildrel<\over
{\smash{\scriptstyle\approx}\vphantom{_2}}$}} {\left( {\frac{K}{\beta }} \right)^{n/2}}{e^{\frac{{n\beta }}{{24K}}}}
\end{align}
\item The non-trivial spin structures will change the phase factor in the overall sum. However, in the dominate large $\frac{K}{\beta}$ regime it will only give an overall constant.

\end{itemize}

\section{How symmetry plays with chaos}\label{pred}
Based on the analysis above, we could obtain some predictions using an effective action that is simply combined from a Schwarzian theory and a \emph{particle on a group} theory, which is called $\text{Sch}_G$.

\subsection{Form factors}
Spectral form factor is an important quantity to quantify the discreteness of the spectrum in a random systems, which could be useful to understand properties of quantum gravity in the black hole and information scrambling in the quantum many-body system. For instance, the two point form factor is defined as the product of the analytic-continued partition function
\begin{align}
{\mathcal{R}_2}(\beta ,t) = \left\langle {Z(\beta  + it)Z(\beta  - it)} \right\rangle 
\end{align}
This quantity is widely studied recently, especially in the context of SYK model (for instance, see \cite{Cotler:2016fpe,Saad:2018bqo}). In the Schwarzian theory of the SYK model, we have the analytic control in the limit
\begin{align}
\vartheta J\equiv \sqrt{\beta^2+ t^2}J \gg 1
\end{align}
In this timescale, we have a specific decaying epoch where (see figure \ref{fig1})
\begin{figure}[htbp]
  \centering
  \includegraphics[width=0.4\textwidth]{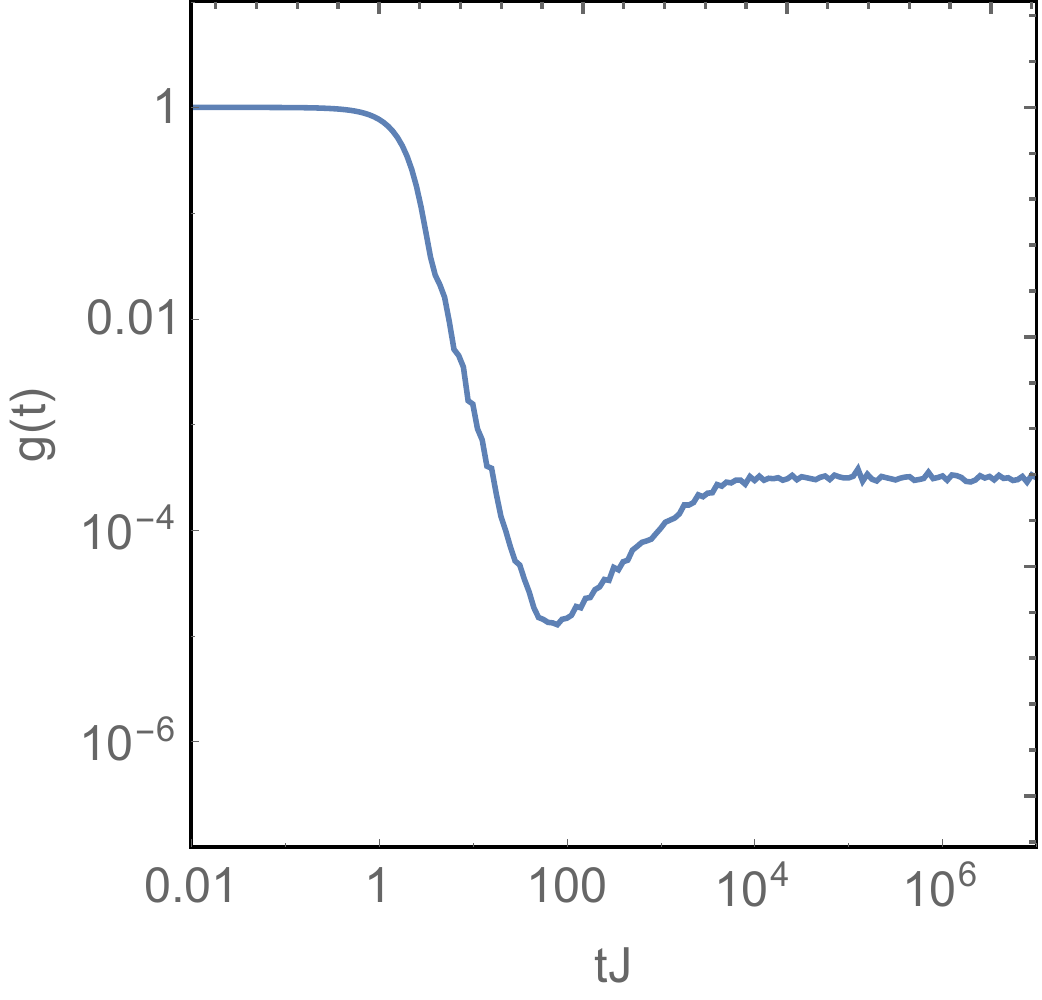}
  \caption{\label{fig1} Example of spectral form factor $\mathcal{R}_2(\beta,t)$ in SYK model. We take $\beta J=1$ and $N=24$ with 800 random realizations.}
\end{figure}
\begin{align}
{\mathcal{R}_2}(\beta ,t) \sim {\left| {\left\langle {Z(\beta  + it)} \right\rangle } \right|^2}
\end{align}

We know that when $\vartheta J \gg N$, there is no contribution from symmetries, and we simply obtain the power law,
\begin{align}
{\mathcal{R}_2}(\beta ,t)\sim \frac{1}{(\vartheta J)^3}
\end{align}
For  $\vartheta J \ll N$ we have the following Table \ref{tb1}.
\begin{table}[h!]
  \begin{center}
    \caption{The spectral form factor for $\text{Sch}_G$ in the window $1\ll \vartheta J \ll N$.}
    \label{tb1}
    \begin{tabular}{c|c|c|c} 
      Group& Number of&{Single sector} &Whole sector \\
      $G$ &spin structures& small $r$ and $\mu$  & general form \\
      \hline
      General $G$ &$\left| {{{\rm{H}}^1}(G,{\mathbb{Z}_2})} \right|$& $\frac{1}{(\vartheta J)^{2p+3}}$  & $\frac{{{e^{\frac{{\# n\vartheta J}}{{N}} - \frac{{\# N}}{{\vartheta J}}{\rm{Mi}}{{\rm{n}}_\mu }{{\left| {\phi  + 2\pi \mu } \right|}^2}}}}}{{{{(\vartheta J)}^{n + 3}}}}$\\
      $\text{U}(1)=\text{SO}(2)$ &$2$& $\frac{1}{(\vartheta J)^{3}}$ & $\frac{{{e^{ - \frac{{\# N}}{{\vartheta J}}{\rm{Mi}}{{\rm{n}}_\mu }{{\left| {\phi  + 2\pi \mu } \right|}^2}}}}}{{{{(\vartheta J)}^4}}}$\\
      $\text{SU}(2)$ & 1 & $\frac{1}{(\vartheta J)^5}$   & $\frac{{{e^{ - \frac{{\# N}}{{\vartheta J}}{\rm{Mi}}{{\rm{n}}_\mu }{{\left| {\phi  + 2\pi \mu } \right|}^2}}}}}{{{{(\vartheta J)}^6}}}$\\
      $\text{SU}(M+1)_{M\ge 1}$ & 1 &$\frac{1}{(\vartheta J)^{M^2+M+3}}$   & $\frac{{{e^{\frac{{\# M(M+2)\vartheta J}}{{N}} - \frac{{\# N}}{{\vartheta J}}{\rm{Mi}}{{\rm{n}}_\mu }{{\left| {\phi  + 2\pi \mu } \right|}^2}}}}}{{{{(\vartheta J)}^{M^2+2M+ 3}}}}$\\    
      $\text{SO}(2M+1)_{M\ge 1}$ & 2 & $\frac{1}{(\vartheta J)^{2M^2+3}}$   & $\frac{{{e^{\frac{{\# (2M+1)M\vartheta J}}{{N}} - \frac{{\# N}}{{\vartheta J}}{\rm{Mi}}{{\rm{n}}_\mu }{{\left| {\phi  + 2\pi \mu } \right|}^2}}}}}{{{{(\vartheta J)}^{2M^2+M + 3}}}}$\\      
      $\text{SO}(2M)_{M\ge 1}$ & 4 & $\frac{1}{(\vartheta J)^{2M^2-2M+3}}$   & $\frac{{{e^{\frac{{\# (2M-1)M\vartheta J}}{{N}} - \frac{{\# N}}{{\vartheta J}}{\rm{Mi}}{{\rm{n}}_\mu }{{\left| {\phi  + 2\pi \mu } \right|}^2}}}}}{{{{(\vartheta J)}^{2M^2-M + 3}}}}$\\     
      $\text{U}(M+1)_{M\ge 1}$ & 2 &$\frac{1}{(\vartheta J)^{M^2+M+3}}$   & $\frac{{{e^{\frac{{\# (M+1)^2\vartheta J}}{{N}} - \frac{{\# N}}{{\vartheta J}}{\rm{Mi}}{{\rm{n}}_\mu }{{\left| {\phi  + 2\pi \mu } \right|}^2}}}}}{{{{(\vartheta J)}^{M^2+2M+ 4}}}}$\\   
      \end{tabular}
  \end{center}
\end{table}
From this table, we obtain predictions of the form factors. We will comment on this in the following:
\begin{itemize}
\item The group $G$, if associating with SYK model, will provide extra fruitful dynamics in form factors. It is relatively easy to observe it with a clean decaying rate in the single charge sector, where the number of positive roots in the group $G$ will provide a contribution and make the decay procedure faster. It is also possible to observe it in the result of the form factor in the whole sector, where the dimension of the group $G$ will contribute and make the decay faster.
\item An exponential decay will happen in the single charge sector if a component of the sector $\mu_i$ is sufficiently large, and we have
\begin{align}
\mathcal{R}_2\sim \exp\left(-\frac{\#\mu^2\vartheta}{K}\right)
\end{align}
\item The spin structures $\sigma$ generically are hard to change the scaling of the spectral form factor. 

\end{itemize}

\subsection{Partition function and density of states}
We briefly comment on the thermodynamical implications of the partition function result in this section. 

Generically the partition function over the given chemical potential $\phi$ is given by
\begin{align}
{Z }(\beta |\phi ) = {\rm{Tr}}({e^{ - \beta H + i\phi Q}})
\end{align}
where $Q$ is the charge operator. In the grand canonical ensemble, the density matrix is given by
\begin{align}
\rho_\text{den}(\beta|\phi)  = \frac{{{e^{ - \beta H + i\phi Q}}}}{{{Z_\sigma }(\beta |\phi )}}
\end{align}
One can do the low temperature expansion of the partition function, and we obtain
\begin{align}
{Z }(\beta |\phi ) = \exp \left( { - \beta {H_0}(\phi ) + i\phi {Q_0}(\phi ) + {S_0}(\phi ) + \frac{{c_0(\phi )}}{{2\beta }} + {\rm{corrections}}} \right)
\end{align}
where for observable $X$, we define
\begin{align}
{X_0} = \operatorname{Tr}(X\rho (\beta  = 0|\phi ))
\end{align}
Namely, $H_0$, $S_0$, $Q_0$, $c_0$ are energy, entropy, charge, specific heat in the state with chemical potential $\phi$ and zero temperature. 

The last term, \emph{corrections}, is obtained from the effective actions. Generically, we would say that for $\text{Sch}_G$ with dimension $\dim G=n$, we have
\begin{align}
{\rm{corrections}} \sim \frac{{n + 3}}{2}\log \beta J
\end{align}

We perform the Fourier transform to obtain the partition function in the single charge sector by (here we take $\text{U}(1)$ for simplicity)
\begin{align}
Z(\beta ,\mu ) = \int_0^{2\pi } {\frac{{d\phi }}{{2\pi }}} {e^{ - i\phi \mu }}Z(\beta |\phi )
\end{align}
In canonical ensemble, we have a similar low temperature expansion in the single charge sector
\begin{align}
Z(\beta ,\mu ) = \exp \left( { - \beta {H_\mu } + {S_\mu } + \frac{{{c_\mu }}}{{2\beta }} + {\rm{corrections'}}} \right)
\end{align}
where we define $X_\mu$ to be the operator $X$ in zero temperature, and charge sector $\mu$. Namely, $H_\mu$, $S_\mu$, $Q_\mu$, $c_\mu$ are energy, entropy, charge, specific heat in the state with charge sector $\mu$ and zero temperature. For corrections, in the single charge sector case, for $\text{U}(1)$ we have
\begin{align}
{\rm{corrections'}} \sim \frac{{3}}{2}\log \beta J
\end{align}
which is the same as the Schwarzian theory, while in general, we have
\begin{align}
{\rm{corrections'}} \sim \frac{{2p+3}}{2}\log \beta J
\end{align}
where $p=|R_+|$, the number of positive roots in $G$. 

In general, the quantities $X_0(\phi)$ and $X_\mu$, could be computed by numerical analysis. We look forward to seeing those developments in the future. 

There is another interesting thermodynamical observable we could look at, which is the density of states. The density of states is given by the Laplace transform of the temperature
\begin{align}
&\rho (E|\phi ) \sim \frac{1}{{2\pi i}}\int_{\gamma  + i\mathbb{R}}^{} {d\beta Z(\beta |\phi )\exp (\beta E + \frac{\# }{\beta })} \nonumber\\
&\rho (E,\mu ) \sim \frac{1}{{2\pi i}}\int_{\gamma  + i\mathbb{R}}^{} {d\beta Z(\beta ,\mu )\exp (\beta E + \frac{\# }{\beta })} 
\end{align}
where $\gamma$ is an arbitrary real constant. Using saddle point approximation, we know that in general, setting $Z\sim \beta^{-\alpha}$, for small $E$ we have
\begin{align}
&\int {d\beta \frac{1}{{{\beta ^\alpha }}}\exp (\beta E + \frac{c}{\beta })} \nonumber\\
&\sim {\left( {\frac{{eE}}{\alpha }} \right)^\alpha }\int {d\beta \exp (\frac{1}{2}\frac{{{E^2}}}{\alpha }{{\beta}^2})}  \nonumber\\
&\sim {E^{\alpha  - 1}}
\end{align}
while for large $E$ we have
\begin{align}
&\int {d\beta \frac{1}{{{\beta ^\alpha }}}\exp (\beta E + \frac{c}{\beta })} \nonumber\\
&\sim {\left( {\frac{E}{c}} \right)^{\alpha /2}}{e^{2\sqrt {cE} }}\int {d\beta \exp (\frac{1}{2}\frac{{{E^{3/2}}}}{{\sqrt c }}{{(\beta  - {\beta _0})}^2})} \nonumber\\
&\sim {e^{2\sqrt {cE} }}{E^{\alpha /2 - 3/4}}
\end{align}
Thus, with our previous result, we know that for small $E$ we have
\begin{align}
&\rho (E|\phi ) \sim {(EJ)^{n/2+1/2}}\nonumber\\
&\rho (E,\mu ) \sim {(EJ)^{p+1/2}}
\end{align}
while for large $E$ we have
\begin{align}
&\rho (E|\phi ) \sim {(EJ)^{n/4}}\nonumber\\
&\rho (E,\mu ) \sim {(EJ)^{p/2}}
\end{align}
For reader's convenience, we will list the energy dependence on $E$ in the following Table \ref{tb2}.
\begin{table}[h!]
  \begin{center}
    \caption{The density of states for $\text{Sch}_G$}
    \label{tb2}
    \begin{tabular}{c|c|c|c|c} 
      Group $G$& {Single sector}  &Whole sector & {Single sector}  &Whole sector\\
      & small $E$ & small $E$ & large $E$ & large $E$ \\
      \hline
      General $G$&$(EJ)^{p+1/2}$&$(EJ)^{n/2+1/2}$& $(EJ)^{p/2}$& $(EJ)^{n/4}$\\       
      $\text{U}(1)=\text{SO}(2)$&$(EJ)^{1/2}$&$(EJ)^{1}$&1 & $(EJ)^{1/4}$ \\
      $\text{SU}(2)$ &$(EJ)^{3/2}$&$(EJ)^{2}$&  $(EJ)^{1/2}$ & $(EJ)^{3/4}$\\     
       $\text{SU}(M+1)_{M\ge 1}$ &$(EJ)^{(M^2+M+1)/2}$&$(EJ)^{(M+1)^2/2}$& $(EJ)^{(M^2+M)/4}$ &$(EJ)^{(M^2+2M)/4}$\\ 
      $\text{SO}(2M+1)_{M\ge 1}$ &$(EJ)^{(2M^2+1)/2}$&$(EJ)^{(2M^2+M+1)/2}$& $(EJ)^{M^2/2}$&$(EJ)^{(2M^2+M)/4}$\\      
      $\text{SO}(2M)_{M\ge 1}$&$(EJ)^{(2M^2-2M+1)/2}$&$(EJ)^{(2M^2-M+1)/2}$&   $(EJ)^{(M^2-M)/2}$&$(EJ)^{(2M^2-M)/4}$\\  
      $\text{U}(M+1)_{M\ge 1}$&$(EJ)^{(M^2+M+1)/2}$&$(EJ)^{(M^2+2M+2)/2}$&$(EJ)^{(M^2+M)/4}$ &$(EJ)^{(M+1)^2/4}$\\ 
      \end{tabular}
  \end{center}
\end{table}

\subsection{A short comment on thermodynamics}
Here we briefly discuss other thermodynamical quantities of the theory. Here we will focus on the canonical ensemble. In our current language, the free energy in the thermodynamical limit is defined by 
\begin{align}
dF = dU - \frac{1}{\beta }dS + \frac{i}{\beta }\phi d\mu 
\end{align}
where here $S$ is the entropy, $U$ is the internal energy, and $\phi d\mu$ is understood as the inner product over lattice vectors in general. Thus, the chemical potential in equilibrium is defined by 
\begin{align}
\phi  =  - i\beta {\left( {\frac{{\partial F}}{{\partial \mu }}} \right)_\beta }
\end{align}
where the partial derivative here is understood as derivatives on each component of $\mu$. One might also define the grand potential
\begin{align}
\Omega (\phi ,\beta ) = F(\phi ,\beta ) - i\phi \mu (\phi ,\beta )
\end{align}
In the above discussions, we know that 
\begin{align}
Z(\beta ,\mu ) = \exp \left( { - \beta {H_\mu } + {S_\mu } + \frac{{{c_\mu }}}{{2\beta }} + {\rm{corrections'}}} \right)
\end{align}
wherein this work we show that the correction terms are generically logarithmically depending on $\beta J$. The free energy is given by
\begin{align}
F(\beta ,\mu ) =  - \frac{1}{\beta }\log Z(\beta ,\mu ) = {H_\mu } - \frac{1}{\beta }{S_\mu } - \frac{{{c_\mu }}}{{2{\beta ^2}}} + \frac{{{\rm{corrections'}}(\beta )}}{\beta }
\end{align}
in the zero temperature expansion. And from those quantities we could predict other thermodynamical quantities. 
\\
\\
In the above expressions, the terms $H_\mu$, $S_\mu$ and $c_\mu$ depend on, generically, the model itself, while the correction is from the Schwarzian theory $\text{Sch}_G$, the low energy effective action that describes the conformal symmetry breaking in the SYK-like models. In the work \cite{Fu:2016vas}, it is discovered that $H_\mu$ is not universal, while $S_\mu$ is universal in the complex SYK model; namely, it only cares about the scaling dimension and the IR information, without high energy details. The observation of universality is definitely, an important ingredient that is from the property of conformal symmetry of the SYK-like models. In our work, since we only compute the correction terms that are logarithmically depending on the temperature (where the logarithmic piece is generically the effect of symmetry comes in), the fact of universality is not affected by the higher symmetries. Thus, we expect that for SYK-like models, the universality property of the zero temperature entropy should stay the same. The correction terms also cannot encounter the expression of the chemical potential, since there is generically no dependence for the $\text{Sch}_G$ in the single charge sector, on the charge itself. Thus, the chemical potential is dominated by conformal contributions. 
\\
\\
Since the Schwarzian theory directly describes the perturbations above the saddle point, we will expect that the theory $\text{Sch}_G$ will be directly related to the physical quantities that are directly related to perturbations. For instance, the susceptibility matrix of complex SYK model is studied in \cite{Fu:2016vas}, which is directly related to two-point correlation functions of the phase modes and the Schwarzian modes. We expect that we will have similar situations in the model with more general symmetries. Those physical quantities are again model dependent, which will be beyond the scope of this work. We leave those studies for future research.

\subsection{Lyapunov exponents}
A crucial fact of the SYK-like models is that when computing the out-of-time-ordered four-point function, the result will have a Lyapunov growth during the early period. The Lyapunov exponent saturates the chaos bound by Maldacena-Shenker-Stanford \cite{Maldacena:2015waa}
\begin{align}
\lambda_L\le \frac{2\pi}{\beta}
\end{align}
Namely, their chaotic features are maximal. This fact indicates a possible holographic dual of those models (see \cite{Shenker:2013pqa,Shenker:2013yza,Shenker:2014cwa} for reference).
\\
\\
We argue that the one-dimensional SYK model with global symmetries will still have the maximal chaotic exponent.
\begin{itemize}
\item The Schwarzian term in the effective action indicates a reparametrization symmetry in our theory, which means that a dimension two operator $h=2$ will appear in the conformal partial wave expansion of the four-point function. The $h=2$ will create a maximal Lyapunov growth in the chaotic regime. Since the Maldacena-Shenker-Stanford bound says that the $\frac{2\pi}{\beta}$ Lyapunov exponent is maximal, other contributions are not possible to increase the $h=2$ contribution.
\item Like $\text{U}(1)$, the charge operator has dimension $h=0$, that means that by shadow transformation $h\to 1-h$ we have $h=1$ contribution appear in the four-point function expansion. However, $h=0$ could never contribution any chaotic behavior. This argument is completely presented in the $\text{U}(1)$ case, see \cite{Bulycheva:2017uqj}.
\end{itemize}
As a conclusion, we expect that generically, the one-dimensional SYK model, attaching with a global symmetry, should still be maximally chaotic. A more detailed analysis of this point is left for future work. For concrete evidence, see for instance, \cite{Gross:2016kjj,Bulycheva:2017uqj,Yoon:2017nig}.
\section{Conclusion}\label{conc}
In this paper, we study various aspects of partition functions for free theory on the symmetry group $G$, and its implications on the SYK model and chaotic dynamics. Symmetry group $G$ will provide charge sectors in the Hamiltonians of theory, and thus allow us to define the canonical and grand canonical ensembles. We study behaviors of partition functions in different ensembles, namely, partition function in specific charge sectors and specific chemical potentials, and claim that those behaviors will affect some chaotic observables and related thermodynamics. For instance, those symmetry groups will generically make the scrambling faster, observed in the spectral form factors. 

Some possible future directions could be given as the following,
\begin{itemize}
\item It would be interesting to generalize formally how partition function behaves in the non-semisimple groups. 
\item It would be interesting to construct specific models corresponding to those symmetry classes and verify their behavior, analytically and numerically. 
\item It would be interesting to study more details about thermodynamics in single charge sectors, and importantly, using Schwarzian theory to compute correlation functions and make predictions in condensed matter systems, for instance, properties of thermoelectric transport.
\item It would be interesting to understand the meaning of those results in the dual gravity. Traditionally, people believe that global symmetry in CFT could be dual to gauge symmetry in AdS. One may address the dual gravitational theory of SYK-like models using the predictions from this paper.

\end{itemize}

\section*{Acknowledgments}
We thank Jordan Cotler, Dionysios Anninos, Nicholas Hunter-Jones, Davide Gaiotto, Yingfei Gu, John Preskill, Subir Sachdev, David Simmons-Duffin, Douglas Stanford, Beni Yoshida and Yi-Zhuang You for discussions and suggestions. JL especially thanks Douglas Stanford for his help on understanding the one-loop determinants of $\mathcal{N}=2$ supersymmetric SYK model and complex SYK models, and comments on the charge sectors. JL also acknowledges Yingfei Gu for his valuable discussions on the charge sectors of complex SYK models. JL is supported in part by the Institute for Quantum Information and Matter (IQIM), an NSF Physics Frontiers Center (NSF Grant PHY-1125565) with support from the Gordon and Betty Moore Foundation (GBMF-2644), and by the Walter Burke Institute for Theoretical Physics. YZ is supported by the graduate student program at the Perimeter Institute.

\end{document}